\documentclass[12pt]{article}
\usepackage{amsmath,amssymb,fullpage}

\usepackage{enumerate}

\usepackage{graphicx}

\usepackage{makeidx}
\usepackage[utf8]{inputenc}
\usepackage{wrapfig}
\usepackage{amsmath}
\usepackage{amsfonts}
\usepackage{mathrsfs}
\usepackage{amssymb}
\usepackage{latexsym}
\usepackage{amsbsy}
\usepackage{color}
\usepackage{bbm}

\usepackage{mathrsfs}
\usepackage{wasysym}
\providecommand{\otherindexspace}[1]{}

\usepackage{amsthm}

\newtheorem{theorem}{Theorem}[section]

\newtheorem{lemma}[theorem]{Lemma}
\newtheorem{proposition}[theorem]{Proposition}
\newtheorem{remark}[theorem]{Remark}

\newtheorem{definition}[theorem]{Definition}

\newtheorem{example}[theorem]{Example}
\newtheorem{assumption}[theorem]{Assumption}

\numberwithin{equation}{section}

\def\cal#1{\mathcal{#1}}

\def \H{\mathbb {H}}
\def \R{\mathbb {R}}

\usepackage{fancyhdr}
\usepackage{graphics}
\usepackage{graphicx}

\DeclareGraphicsExtensions{.eps,.bmp,.jpg,.pdf,.mps,.png,.gif}

\makeatletter
\def\titre{\@title}
\makeatother

\title{Game options in an imperfect market with default}

\author{Roxana Dumitrescu\thanks{Department of Mathematics, King's College London, United Kingdom,  email: \textbf{roxana.dumitrescu@kcl.ac.uk}} \and Marie-Claire Quenez \thanks{LPMA,
Université Paris 7 Denis Diderot, Boite courrier 7012, 75251 Paris cedex 05, France, email: \textbf{quenez@math.univ-paris-diderot.fr}} \and  Agnès Sulem
\thanks{ INRIA Paris,  3 rue Simone Iff, CS 42112, 75589 Paris Cedex 12, France, email: \textbf{agnes.sulem@inria.fr} The authors thank an anonymous referee for insightful comments which led to a significantly improved version of the paper.}}

\begin{document}

\date{\today}

\maketitle

\begin{abstract}

We study  pricing  and  superhedging strategies for game options in an imperfect market 
with default. 
%with contraints, when the probability of a default is strictly positive.
We extend the results obtained by Kifer in \cite{Kifer}  in the case of a perfect market model to the case of an imperfect  market with default,  when the imperfections are taken into account via the nonlinearity  of the wealth dynamics. 
%In this framework, the  pricing system is expressed as a nonlinear $g$-expectation/evaluation induced by a nonlinear BSDE with default jump.
We  introduce the {\em seller's price} of the game option as the infimum of the initial wealths which allow the seller to be superhedged.
We {prove} that this price coincides with the value function of an associated {\em generalized} Dynkin game, recently introduced in \cite{DQS2},  expressed with a nonlinear expectation induced by a nonlinear BSDE with default jump. We moreover 
study the existence of superhedging strategies. 
 We then address  the case of ambiguity on the model, - for example ambiguity on the default probability - and characterize  the robust seller's price of a game option as the value function of a {\em mixed generalized} Dynkin game.
 We study the existence of a cancellation time  and a trading strategy which allow the seller to be super-hedged, whatever the model is. 
% This study is introduced by the analysis of the simpler case of American options. 
 \end{abstract}


%\end{titlepage}

\section{Introduction}

%Numerous works have addressed stochastic control,  optimal stopping 

Game options, which have been introduced by Kifer (2000) \cite{Kifer}, are derivative contracts that can be terminated by both counterparties at 
any time before a maturity date $T$. More precisely, a game option allows the seller to cancel it and the buyer to exercise it at any stopping time smaller than $T$. 
%If the buyer (resp. the seller)  exercises (resp. cancels) at maturity time $T$, then   the seller pays the amount $g(X_{T})$  to the buyer.
 If the buyer exercises at time $\tau$ before the seller cancels, then the seller pays the buyer  the amount $ \xi_{\tau}$, but if the seller 
cancels  before the buyer exercises, then he pays  the amount $\zeta_{\sigma} \geq  \xi_{\tau}$ to the buyer at the cancellation time $\sigma$. The difference 
$\zeta_{\sigma} - \xi_{\sigma}$ is interpreted as a penalty that the seller pays to the buyer for the cancellation of the contract. 
%In other terms, if the seller (resp. the buyer)  selects a cancellation time $\sigma$ (resp. an exercise time $\tau$),  the seller pays to the buyer at time 
%$\tau \wedge \sigma$
%the {\em payoff}
%If the buyer exercises the contract at time $t$ then he gets  the payment $\xi_{t}$ from the seller, but if the latter cancels before $\tau$, then he has to pay the {amount} $\zeta_{t}$ to the buyer. The difference $\delta_t:=\zeta_t-\xi_t$ is {assumed to be} non negative for all $t$ and is interpreted as a penalty that the seller pays to the buyer for cancellation of the contract. 
 In short, if the buyer selects an exercise time $\tau $ and the seller selects a cancellation time $\sigma$, 
 then the latter pays to the former the payoff
 $\xi_\tau \textbf{1}_{\tau \leq \sigma}+ \zeta_\sigma \textbf{1}_{\tau >\sigma}$ at time $\tau \wedge \sigma$.

%as a { game version} of American options giving the additional  possibility to the option seller ({\bf also called }writer or  issuer) to cancel { the option at any time,} paying for this a prescribed penalty. 
%More precisely, a game contingent claim is a contract between a seller $S$ and a buyer $B$ which { allows}  $S$ to cancel it and $B$ to exercise it at any time $t$ up to a maturity date  $T$ when the contract is in any case terminated. If $B$ exercises the contract at time $t$ then he gets  the payment $\xi_t$ from $S$, but if $S$ cancels before $B$ then $S$ has to pay the {amount} $\zeta_t$ to $B$ . If $S$ cancels and $B$ exercises simultaneously at time $t$, then $S$ pays  $\xi_t$ to $B$ . The difference $\delta_t:=\zeta_t-\xi_t$ is {assumed to be} positive for all $t$ and is interpreted as a penalty that $S$ pays to $B$ for cancellation of the contract.
 
 In the case of classical perfect markets, Kifer introduces 
 the "{fair price}"  of the game option, defined as  the minimum initial wealth needed for the seller  to cover his liability to pay the payoff to the buyer  until a cancellation time, whatever is the exercise time chosen by the buyer.
He shows  both in the CCR discrete-time model and in the Black and Scholes model 
%(with $\xi$ and $\zeta$ continuous),
 that this price  is 
  equal  to the value function of the following  Dynkin game:%  introduced by Bismut \cite{Bismut},
\begin{equation}\label{Kifer}
 \sup_\tau \inf_\sigma \mathbb{E}_{Q}[ {\tilde \xi}_\tau \textbf{1}_{\tau \leq \sigma}+ {\tilde \zeta}_\sigma \textbf{1}_{\tau >\sigma}]= \inf_\sigma \sup_\tau \mathbb{E}_{Q}[{\tilde \xi}_\tau \textbf{1}_{\tau \leq \sigma}+ {\tilde \zeta}_\sigma \textbf{1}_{\tau >\sigma}],
\end{equation}
where $\tilde \xi_t$ and $\tilde \zeta_t$ are the discounted values of $\xi_t$ and $\zeta_t$, equal to $e^{-rt} \xi_t$ and 
$e^{-rt} \zeta_t$ respectively in the Black and Scholes model, where $r$ is the instantaneous interest rate.
Here, $\mathbb{E}_{Q}$ denotes the expectation under the unique martingale probability measure $Q$ of the market model. 
Further research on the pricing of game options and on more sophisticated game-type financial contracts includes in particular papers by Dolinsky and Kifer (2007) \cite{DK}  and Dolinsky and al. (2011) \cite{Dolinsky} in the discrete time case, and  by Hamad\`{e}ne (2006) \cite{H} 
in a continuous time perfect market model with continuous payoffs $\xi$ and $\zeta$.
%studied the link between the 
%the price of game options and {\em linear} doubly reflected BSDEs in a Brownian perfect market model with 
%continuous payoffs. 
 We also mention the paper by Bielecki and al. (2009) \cite{BCJR} which studies the pricing of game options in a market model with default.
Note that  in \cite{KK}, Kallsen and Kuhn (2004) study game options in an incomplete market.
They consider
%the concept of utility-based indifference pricing  
another type of pricing called {\em neutral valuation} via utility maximization.

The aim of the present paper is to study pricing and hedging issues for game options in 
 the case of  imperfections in the market model taken into account via the nonlinearity  of the wealth dynamics, modeled via a nonlinear driver $g$. We moreover include the possibility of a default. A large class of imperfect market models can fit in our framework, like different borrowing and lending interest rates, or taxes on the profits from risky investments. 
 Our model also includes the case when the seller of the option is a 
 "large trader" whose  hedging  strategy may affect the market prices and the default probability. 
 
%{\bf Note that in the literature, many authors point out that option price distortions can be induced by 
%feedback effects from hedging strategies (see \cite{PlatenS}+ A COMPLETER).}
%We also stress that this phenomenon can increase due to the imperfections of the market and/or 
%the presence of a default. 
%In this setting, the pricing system is expressed as a nonlinear expectation/evaluation $\mathcal{E}^g$ induced by a nonlinear BSDE with default jump (solved 
%under the primitive probability measure $P$) with driver $g$.%is determined by the dynamics of the wealth process. 
%{\bf 
%Note that  in \cite{KK}, Kallsen and Kuhn (2004) study game options in an incomplete market.
%They consider
%%the concept of utility-based indifference pricing  
%another type of pricing called {\em neutral valuation} via utility maximization, and use a different approach.}
 Here, we suppose that  the payoffs $\xi$ and $\zeta$ associated with the  game option are  right-continuous left-limited  (RCLL) only and they satisfy Mokobodzki's condition.  We  call {\em seller's price} of the game option,  the infimum (denoted by $u_0$) of the initial wealths such that  there exists  a cancellation time $\sigma$  and a portfolio strategy which allow  the seller to pay $\xi_\tau$ (at time $\tau$)  to the buyer if the buyer exercises at any time ${\tau \leq \sigma}$, and $\zeta_\sigma$ (at time $\sigma$) if the buyer has still not exercise at 
time $\sigma$. Note that this infimum is not necessarily attained. 
 We provide a  characterization of
 the seller's price $u_0$ of the game option 	as the (common) value of a corresponding {\em generalized}
  Dynkin game (recently introduced in \cite{DQS2}). 
 More precisely, we show that
\begin{align}\label{GD}
u_0= \sup_\tau \inf_\sigma \mathcal{E}^g[\xi_\tau \textbf{1}_{\tau \leq \sigma}+ \zeta_\sigma \textbf{1}_{\tau >\sigma}]= \inf_\sigma \sup_\tau \mathcal{E}^g[\xi_\tau \textbf{1}_{\tau \leq \sigma}+ \zeta_\sigma \textbf{1}_{\tau >\sigma}], 
\end{align}
where $\mathcal{E}^g$ is a nonlinear expectation/evaluation induced by a nonlinear BSDE with default jump solved under the primitive probability measure $P$ with driver $g$.
  Note that in the particular case of a perfect market, the driver $g$ is linear and 
  one can show by using an actualization procedure and a change of probability measure that 
 %with continuous payoffs, 
 \eqref{GD} corresponds to $\eqref{Kifer}$. 
 
We prove that, under an additional left-regularity assumption on $\zeta$ (but not on $\xi$), there exist a  cancellation time and a trading strategy which allow the seller to be super-hedged. %, and $u_0$ is called the {\em superhedging price}.
In this case, the infimum in the definition of the seller's price $u_0$ is  attained. 
When $\zeta$ is only RCLL, the infimum is not necessarily attained. 
 However, 
we show that for each $\varepsilon>0$, the amount $u_0$  allows the seller 
to be super-hedged up to $\varepsilon$ until a well chosen cancellation time.  The proofs of these results rely on the links %(established in \cite{DQS2})
between {\em generalized} Dynkin games and nonlinear doubly reflected BSDEs with default jump.

 The second main question  we  study is the pricing and  superhedging problem of game options in the case of uncertainty on 
the (imperfect) market model. To the best of our knowledge, this problem has not been studied in the literature except by Dolinsky (2014) in \cite{Do} in a discrete time framework.
In particular, our model can take into account an {\em ambiguity on the default probability} as illustrated in Section \ref{example}.
We prove that the robust seller's  price of the game option  under uncertainty, defined as the infimum of the initial wealths with allow the seller to be superhedged whatever the model is,  coincides with the value function of a mixed generalized Dynkin game. We also study the existence of robust superhedging strategies.

The paper is organized as follows:  in Section \ref{sec2}, we introduce our imperfect market model with default and nonlinear wealth dynamics.
%Before turning to game options,  we start by the  simpler case when there is no  possibility
%  for the seller  to cancel the option, which corresponds to  American options, and 
%  provide some new results on pricing and hedging issues. 
%  The next sections are the main parts of  the paper. 
In Section \ref{sec3}, we study  pricing and superhedging of game options and their links with {\em generalized Dynkin games}. In Section  \ref{mixed}, we
address the  case of an imperfect market with model ambiguity. %, and we give an application to a model including some ambiguity on the default probability. 
Section \ref{sec-comp} provides some complementary results concerning the buyer's point of view and the case with 
dividends.  Some results on doubly reflected BSDEs with default jumps and a useful  lemma of analysis  are given in  Appendix.

\section{Imperfect market model with default}\label{sec2}

\subsection{Market model with default}\label{marketmodel}
Let $(\Omega, \mathcal{G}, {P})$ be a complete probability space 
 %We assume that all processes are defined on a finite time horizon $[0,T]$,  with $T < \infty$,  and we suppose the space to be 
 equipped with two stochastic processes:
  a unidimensional standard Brownian motion $W$ and a jump process $N$ defined by 
  $N_t={\bf 1}_{\vartheta\leq t}$ for any $t\in[0,T]$, where $\vartheta$ is a random variable which models a default time. We assume that this default can appear at any time that is $P(\vartheta \geq t)>0$ for any $t\geq 0$. We denote by ${\mathbb G}=\{\mathcal{G}_t, t\geq 0 \}$ the {\em augmented filtration} that is generated by $W$ and $N$ (in the sense of \cite[IV-48]{DM1}). We suppose that  $W$ is a ${\mathbb G}$-Brownian motion. 
 We denote by ${\cal P}$ the ${\mathbb G}$-predictable $\sigma$-algebra.
 Let  $(\Lambda_t)$ be the  predictable compensator of the nondecreasing process $(N_t)$.
 Note that $(\Lambda_{t \wedge \vartheta})$ is then the predictable compensator of
  $(N_{t \wedge \vartheta} )= (N_t)$. By uniqueness of the predictable compensator, 
  $\Lambda_{t \wedge \vartheta} = \Lambda_t$, $t\geq0$ a.s.
  We assume that $\Lambda$ is absolutely continuous w.r.t. Lebesgue's measure, so that there exists a nonnegative process $\lambda$, 
 called the intensity process, such that $\Lambda_t=\int_0^t \lambda_s ds$, $t\geq0$.
  Since $\Lambda_{t \wedge \vartheta} = \Lambda_t$,  $\lambda$ vanishes after $\vartheta$. 
We denote by $M$ the compensated martingale   which satisfies 
\begin{equation*}
%\label{M}
M_t  = N_t-\int_0^t\lambda_sds\,.
\end{equation*}

%Recall that $\lambda_t$ can be interpreted as the first order approximation of the default probability conditioned by 
%${\cal G}_t$, that is 
%$$\lambda_t = \lim_{h \downarrow 0 } \frac{1}{h} P (t < \vartheta \leq t+h \, / \, {\cal G}_t).$$
Let $T >0$ be the finite horizon. We introduce the following sets:
\begin{itemize}
%\item ${S}^{+,\infty}$ is the set of positive ${\cal G}$-adapted $P$-essentially bounded rcll processes.
\item ${S}^{2}$ 
is the set of ${\mathbb G}$-adapted RCLL processes $\varphi$ such that $\mathbb{E}[\sup_{0\leq t \leq T} |\varphi_t | ^2] < +\infty$.
\item ${\cal A}^2$  is the set of real-valued non decreasing RCLL predictable
 processes $A$ with $A_0 = 0$ and $\mathbb{E}(A^2_T) < \infty$.

%\item $L^{1,+}$ is the set of positive ${\cal G}$-adapted RCLL processes $\varphi$ such that $E[\varphi_t]<\infty$ for any $t\in [0,T]$.
\item ${\mathbb H}^2$  is the set of ${\mathbb  G}$-predictable processes $Z$ such that
 $
 \| Z\|^2:= \mathbb{E}\Big[\int_0^T|Z_t|^2dt\Big]<\infty \,.
 $
\item  ${\mathbb H}^2_{\lambda}:= L^2( \Omega \times [0,T],{\cal P}, \lambda_tdt)$, equipped with the scalar product $\langle U,V \rangle _{\lambda}:= \mathbb{E}\Big[\int_0^TU_t V_t\lambda_tdt\Big]$, for all 
$U,V$ in ${\mathbb H}^2_{\lambda}$. For each $U \in$ ${\mathbb H}^2_{\lambda}$, we set 
$\| U\|_{\lambda}^2:=\mathbb{E}\Big[\int_0^T|U_t|^2\lambda_tdt \Big]<\infty \,.$
\end{itemize}

%Recall that a process $\phi = (\phi_t)_{0 \leq t \leq T}$ is called  {\em predictable} if it is
% ${\mathbb G}$-predictable (or equivalently 
 %${\mathcal P}$-measurable. 

Note that for each $U \in {\mathbb H}^2_{\lambda}$, we have $\| U\|_{\lambda}^2=\mathbb{E}\Big[\int_0^{T\wedge \vartheta}|U_t|^2\lambda_tdt\Big] \,$ because 
the ${\mathbb G}$-intensity $\lambda$ vanishes after $\vartheta$. 
Moreover, we can suppose that for each 
$U$ in ${\mathbb H}^2_{\lambda}$ $=$ $L^2( \Omega \times [0,T],{\cal P}, \lambda_tdt)$, $U$ (or its representant in ${\cal L}^2( \Omega \times [0,T],{\cal P}, \lambda_tdt)$ still denoted by 
$U$) vanishes  after $\vartheta$. 
\\
Moreover,  $\mathcal{T}$ denotes the set of
stopping times $\tau$ such that $\tau \in [0,T]$ a.s.\, and for each $S$ in $\mathcal{T}$, 
   $\mathcal{T}_{S}$ is  the set of
stopping times
$\tau$ such that $S \leq \tau \leq T$ a.s.

We recall the martingale representation theorem 
(see e.g.  \cite{JYC}):
\begin{lemma}\label{theoreme representation}
Any ${\mathbb  G}$-local martingale $m= (m_t)_{0\leq t \leq T}$ has the representation 
\begin{equation}
\label{equation representation}
m_t  = m_0+\int_0^t z_sdW_s+\int_0^t l_sdM_s \,, \quad\forall\,t\in[0,T] \quad a.s. \,,
\end{equation}
where $z= (z_t)_{0\leq t \leq T}$ and $l= (l_t)_{0\leq t \leq T}$ are %${\cal G}$-
predictable such that the two above stochastic integrals are well defined.
If $m$ is a square integrable martingale, then $z \in {\mathbb H}^2$ and $l \in {\mathbb H}^2_{\lambda}$.
\end{lemma}

 We consider now a financial market with three assets with price process $S=(S^{0}, S^{1},S^{2})'$
governed by the equation:
\begin{equation*}
\begin{cases}
dS_t^{0}=S_t^{0} r_tdt\\
dS_t^{1}=S_t^{1}[\mu_t^1dt +  \sigma^1_tdW_t]\\
dS_t^{2}=S_{t^-}^{2} [\mu^2_tdt+\sigma^2_tdW_t-dM_t].
\end{cases}
\end{equation*}
The process $S^0= (S_t^{0})_{0\leq t \leq T}$ corresponds to the price of a non risky asset with interest rate 
process $r= (r_t)_{0\leq t \leq T}$, 
$S^1= (S_t^{1})_{0\leq t \leq T}$ to a non defaultable risky asset, and $S^2= (S_t^{2})_{0\leq t \leq T}$ to a defaultable  asset with total default. The price process $S^2$ vanishes after $\vartheta$.
 
All the processes $\sigma^1,\sigma^2,$ $r, \mu^1,\mu^2$ are % ${\mathbb G}$-  
predictable (that is  ${\cal P}$-measurable). 
We set  $\sigma=(\sigma^1,\sigma^2)'$. 
%\beq
%dS_t &=&S_{t^-}(\mu_tdt+\sigma_tdW_t+\beta_tdN_t) \,,
%\label{actif S}
%\enq
We make  the following assumptions:
\begin{assumption}\label{hypo coeff}%\hfill
The coefficients $\sigma^1, \sigma^2 > 0$, 
%\beta \neq 0,$ 
and  $r$, $\sigma^1,\sigma^2,$ $\mu^1,\mu^2,\lambda,$ $\lambda^{-1}$,${(\sigma^1)}^{-1}$, 
${(\sigma^2)}^{-1}$ are bounded. 
%\begin{equation*}
%\int_0^T \Big( | \mu_t | + | \sigma_t | ^2 + \lambda_t | \beta_t |^2 \Big) dt ~< ~\infty \quad a.s.
%\end{equation*}
%\item $\beta$ satisfies $\beta_t > -1$, $0 \leq t \leq T$
% $dt \otimes dP$-a.s.
%\end{enumerate}
%}
\end{assumption}
%The condition (ii) ensures the positivity of the price process $S^2$.\\
We consider an investor, endowed with an initial wealth equal to $x $, who can invest his wealth in the three assets of the market. 
At each time $t < \vartheta$, he chooses   the amount $\varphi_t^1$ (resp. $\varphi_t^2$) of wealth invested in the first 
(resp. second) risky asset. However, after time $\vartheta$, the investor cannot invest his wealth in the defaultable 
asset since its price is equal to $0$, and he only chooses  the amount $\varphi_t^1$ of wealth invested in the first risky asset. Note that the process $\varphi^2$ can be defined on the whole interval $[0,T]$ by setting $\varphi_t^2=0$ for each $t \geq \vartheta$.
%the process $\varphi^2$ vanishes from $\vartheta$. 
A process $\varphi_.= (\varphi_t^1, \varphi_t^2)'_{0 \leq t \leq T}$ is called a {\em risky assets stategy} if 
it belongs to ${\mathbb H}^2 \times  {\mathbb H}^2_{\lambda}$.\\
We denote by  $V^{x, \varphi}_t$ (or simply $V_t$) the {\em wealth},
 or equivalently the value of the portfolio, at 
time $t$. The amount invested in the non risky asset at time $t$ is then given by $V_t - (\varphi_t^1+ \varphi_t^2)$.
%A process $(X_t, \varphi_t^1, \varphi_t^2)_{0 \leq t \leq T}$ is called a {\em wealth/risky assets stategy} if 
%it belongs to $\mathcal{S}^2 \times{\mathbb H}^2 \times {\mathbb H}^2_{\lambda}$ and 
\paragraph{The perfect market model.}
In the classical case of a perfect market model,  the wealth process and the strategy satisfy the self financing condition:
\begin{equation}\label{portfolio}
dV_t  = (r_t V_t+\varphi_t^1 (\mu^1_t - r_t)+\varphi_t^2(\mu^2_t - r_t) ) dt + 
(\varphi_t^1 \sigma^1_t + \varphi_t^2 \sigma^2_t) dW_t - \varphi_t^2  dM_t.
% & = (r_t V_t+\varphi_t^1 \sigma_t^1  \theta_t^1+ \varphi_t^2 (\mu^2_t - r_t)  ) dt + 
%\varphi_t ' \sigma_t dW_t - \varphi_t^2  dM_t,
\end{equation}
%where $\theta_t^1:=\dfrac{\mu_t^1-r_t}{\sigma_t^1}$.\\
Setting $K_t:=- \varphi_t^2$, and $Z_t:=\varphi_t^1 \sigma^1_t + \varphi_t^2 \sigma^2_t$, which implies that 
$\varphi_t^{1} = (Z_t +   \sigma^2_t  K_t)( \sigma_t^1)^{-1}$,
%\frac{Z_t -  \varphi_t^{2} \, \sigma^2_t }{\sigma^1_t}= 
% \frac{Z_t +   \sigma^2_t  K_t\, }{\sigma^1_t}$,  
we get 
\begin{align*}
dV_t & = r_t V_t+(Z_t + { \sigma^2_t} K_t ) (\mu_t^1-r_t)( \sigma_t^1)^{-1}-K_t (\mu^2_t - r_t)  dt + Z_t dW_t + K_t  dM_t\\
 & = (r_t V_t+ Z_t  \theta_t^1+ K_t \theta_t ^2 \lambda_t) dt  
 + Z_t dW_t + K_t  dM_t,
\end{align*}
where $\theta_t^1:=\dfrac{\mu_t^1-r_t}{\sigma_t^1}$ and $\theta_t^2:=  \dfrac{\sigma_t^2 \theta_t^1 -\mu_t^2+r_t}{\lambda_t  }\,{\bf 1}_{\{t \leq \vartheta \} }$. \\
%By the condition (i), the local martingale $(Z_t)_{0 \leq t \leq T}$ defined by 
%$dZ_s= Z_{s^-} [\theta^1_s dW_s + \theta^2_s dM_s]$, with $Z_0=1$, is a martingale.
%Hence,  there exists a martingale probability measure with density $Z_T$ on ${\cal G}_T$, and it is the unique one. 
%The market is thus complete.
Consider a European contingent claim with maturity $T>0$ and payoff $\xi$ which is $\mathcal{G}_T$ measurable, belonging to  ${L}^2$. The problem is to price and hedge this claim by constructing a replicating portfolio. 
From  \cite[Proposition 2.6 ]{DQS4}, there exists an unique process $(X, Z, K) \in \mathcal{S}^2 \times {\mathbb H}^2 \times  {\mathbb H}^2_{\lambda}$ solution of the following  BSDE with default jump:
\begin{equation}\label{portfolio}
- dX_t = \displaystyle -  (r_t X_t+Z_t \theta_t^1+K_t  \theta_t^2 \lambda_t) dt -  Z_t dW_t - K_t  dM_t\,; \quad
X_T=\xi.
\end{equation}
%By Assumption \ref{hypo coeff} The uniqueness of the solution follows from the $\lambda$-{\em Lipschitz} property of the associated driver, due to (i) {\bf VOIR APPENDIX}
The solution $(X, Z, K)$ provides the replicating portfolio. More precisely, 
the process $X$  corresponds to its value, and 
the hedging risky assets stategy  $\varphi \in {\mathbb H}^2_{\lambda}$ is given by $\varphi=\Phi (Z, K)$, 
where $\Phi$ is the one to one map defined on ${\mathbb H}^2 \times  {\mathbb H}^2_{\lambda}$ by:
%\begin{equation*} %\begin{align*}
% {\varphi_t}' \sigma_t = Z_t \;\; ; \;\;
% - \varphi_t^2 = K_t \,,
%\end{equation*}
%where ${\varphi_t}' \sigma_t= \varphi_t^1 \sigma^1_t + \varphi_t^2 \sigma^2_t={\varphi ^1_t} \sigma^1_t + {\varphi^2_t} \,{\bf 1}_{\{t \leq \vartheta \} }\sigma^2_t$ since $\varphi^2_t$ vanishes after $\vartheta$.
%This corresponds to the following  change of variables on ${\mathbb H}^2 \times  {\mathbb H}^2_{\lambda}$: 
\begin{definition}\label{stbis}
Let $\Phi$  be the functional defined by 
$$\Phi:{\mathbb H}^2 \times  {\mathbb H}^2_{\lambda} \rightarrow {\mathbb H}^2 \times  {\mathbb H}^2_{\lambda}; 
(Z, K) \mapsto \Phi (Z, K):= \varphi,$$ where $\varphi= (\varphi^1, \varphi^2)$ is given by 
\begin{equation*}
 \varphi_t^{2} = - {K_t} \;\; ; \;\; 
\varphi_t^{1} = 
 \frac{Z_t +   \sigma^2_t  K_t\, }{\sigma^1_t},
\end{equation*}
which is equivalent to 
$%\begin{align*}
K_t=- \varphi_t^2\, ;\,\,\,
Z_t= {\varphi ^1_t} \sigma^1_t + {\varphi^2_t}\, \sigma^2_t = {\varphi_t}' \sigma_t .
$ 
% \end{align*}
\end{definition}
Note that the processes $\varphi^2$ and $K$, which belong to ${\mathbb H}^2_{\lambda}$,
both vanish after time $\vartheta$.

The process $X$ coincides with $V^{X_0, \varphi}$, the value of the  portfolio 
associated with initial wealth $x=X_0$ and portfolio strategy $\varphi$. 
From the seller's point of view, this portfolio is a hedging portfolio.  Indeed, by investing the initial amount $X_0$ in the reference assets along the strategy $\varphi$, the seller  can  pay the amount $\xi$ to the buyer at time $T$ (and similarly at each initial time $t$). 
We derive that $X_t$ is the price at time $t$ of the option, called {\em hedging price}, and denoted by 
$X_t(\xi)$. 
% first introduced in \cite{EQ96} in a Brownian framework (later called 
%{\em $g$-evaluation} in 
%\cite{Peng2004}) and denoted by ${\cal E}^g$.
%Since the driver of BSDE \eqref{portfolio} is $\lambda$-{\em linear}, 
%(see \eqref{ll} in the Appendix), 
By the representation property of  the solution  of a $\lambda$-linear BSDE with default jump (see \cite[Theorem 2.13]{DQS4}), we have that the solution $X$ of BSDE \eqref{portfolio} can be written as follows:  
\begin{equation}\label{free}
X_t(\xi)=\mathbb{E}[e^{-\int_t ^T r_s ds} \zeta_{t,T}\xi \,|\,{\cal G}_t], 
\end{equation} where 
$\zeta_{t, \cdot}$ satisfies 
\begin{equation}\label{def-zeta}
d\zeta_{t,s}= \zeta_{t,s^-} [-\theta^1_s dW_s - \theta^2_s  dM_s]; \quad  \zeta_{t,t}=1.
\end{equation}
 This defines a {\em linear} price system $X$: $\xi \mapsto X (\xi)$.
Suppose now that 
\begin{equation}\label{cth}
\theta^2_t < 1, \; 0 \leq t \leq \vartheta\, \,dt \otimes dP -a.s.\end{equation}
%Then, by Proposition \ref{roro} and 
%Remark \ref{moinsun}, $\zeta_{0,.}$ is a square integrable positive martingale. 
%By classical results, 
%the probability measure with density $\zeta_{0,T}$ on ${\cal G}_T$ is the unique 
% {\em martingale probability measure}. 
Then $\zeta_{t,\cdot}>0$.
Let $Q$ be the probability measure which admits  $\zeta_{0,T}$ as density on ${\cal G}_T$.
Using Girsanov's theorem, it can be shown that $Q$ is the unique martingale probability measure.%, that is such that the discounted prices of the reference assets are $Q$-martingales.
 In this case, the price system $X$ is increasing 
and corresponds to the classical
 arbitrage free price system (see \cite{JYC, BJR, BCJR}).
%Hence,  there exists a martingale probability measure with density $Z_T$ on ${\cal G}_T$, and it is the unique one. 
%The market is thus complete.

\begin{remark} We have presented above the case of a defaultable asset with total default. 
A different model for the asset price $S^2$ (see  e.g. \cite[Chapter 7, Section 9.3]{JYC}) could be considered: 
$$dS_t^{2}=S_{t^-}^{2} [\mu^2_tdt+\sigma^2_tdW_t+ \beta_t dM_t],$$
where $\beta _t \neq 0$ and $\beta_t > -1$, with $\beta_t$, $\beta_t^{-1}$ bounded. In this case, the price does not vanish after the default time $\vartheta$. 
%To ensure non arbitrage (and completeness of the market), we have to impose the condition: 
We suppose that 
\begin{equation}\label{conditiono}
\dfrac{\mu_t^1-r_t}{\sigma_t^1}\,{\bf 1}_{\{t > \vartheta \} } = \dfrac{\mu_t^2-r_t}{\sigma_t^2} \,\,{\bf 1}_{\{t >\vartheta \} }\,\quad dt \otimes dP- \text{a.s.}
\end{equation}
%The martingale probability measure's density $\zeta_{0,T}$ on ${\cal G}_T$  is then given by $\eqref{zeta}$ with
Let  $\zeta_{0,\cdot}$ be defined by \eqref{def-zeta} with
$
 \theta_t^1=\dfrac{\mu_t^1-r_t}{\sigma_t^1}; \quad 
\theta_t^2=  \dfrac{\mu_t^2-\sigma_t^2 \theta_t^1-r_t}{\beta_t \lambda_t  }\,{\bf 1}_{\{t \leq \vartheta \} }.$
Assume that $\theta^2_t < 1, \; 0 \leq t \leq \vartheta\, \,dt \otimes dP$ -a.s. The assumption \eqref{conditiono} ensures that the probability measure $Q$ with  $\zeta_{0,T}$ as density on ${\cal G}_T$ is 
the unique martingale probability measure.
%\footnote{Conversely, if the market is arbitrage free, then condition \eqref{conditiono} holds. Moreover, $Q$ is the unique martingale probability. }
The arbitrage free price of the contingent claim $\xi$ is given by \eqref{free} and satisfies 
BSDE \eqref{portfolio};
%\begin{equation}\label{portfolio2}
%- dX_t = \displaystyle -  (r_t X_t+Z_t \theta_t^1+K_t \lambda_t \theta_t^2 ) dt -  Z_t dW_t - K_t  dM_t\,; \quad
%X_T=\xi.
%\end{equation}
moreover, the hedging strategy $\varphi= (\varphi^1, \varphi^2)$ is given by: 
$
 \varphi_t^{2} = \frac{K_t}{\beta_t}$ and $
\varphi_t^{1} = \frac{Z_t -  \varphi_t^{2}  \sigma^2_t }{\sigma^1_t}.
%= 
% \frac{Z_t - \frac{ \sigma^2_t}{ \beta_t} K_t }{\sigma^1_t}.
$

\end{remark}

\paragraph{The imperfect market model ${\cal M}^g$. } From 
%A process $(Z,k)$ is called an {\em admissible control} if it belongs to 
%${\mathbb H}^2 \times {\mathbb H}^2_{\lambda}$. The corresponding risky assets-stategy is then given 
%by \eqref{stbis}.\\
 now on,   we assume that  there are  imperfections in the market which are taken into account via 
the {\em nonlinearity} of the
dynamics of the wealth. More precisely, the
dynamics of the wealth $V$ associated with strategy $\varphi=(\varphi^1, \varphi^2)$  can be written via  a {\em nonlinear} 
driver, defined as follows: %Let us first recall the definition of a Lipschitz driver. %in our framework.
%%\begin{definition}[Driver, Lipschitz driver]\label{defd}
%a  function
%%is said to be a {\em driver} if\begin{itemize} \item
%$g: [0,T]  \times \Omega \times \R^3  \rightarrow \R $; 
%$(\omega, t,y, z, k) \mapsto  g(\omega, t,y,z,k) $
%which   is $ {\cal P} \otimes {\cal B}(\R^3) 
%- $ measurable, and such that
%%\item
% $g(.,0,0,0) \in {\mathbb H}^2$.\\
%%\end{itemize}
%A driver $g$ is called a {\em Lipschitz driver} if moreover there exists a constant $ C \geq 0$ such that 
%$dP \otimes dt$-a.s.\,,
%for each $(y_1, z_1, k_1)$, $(y_2, z_2, k_2)$,
%$$|g(\omega, t, y_1, z_1, k_1) - g(\omega, t, y_2, z_2, k_2)| \leq
%C (|y_1 - y_2| + |z_1 - z_2| +   \lambda_t |k_1 - k_2 |).$$
%%\end{definition}
%driver (as in \cite{EQ96}), that is %Let us first recall the definition of a Lipschitz driver. %in our framework.
\begin{definition}[Driver, $\lambda$-{\em admissible} driver]\label{defd}
A  function $g$
is said to be a {\em driver} if\\
$g: [0,T]  \times \Omega \times \R^3  \rightarrow \R $; 
$(\omega, t,y, z, k) \mapsto  g(\omega, t,y,z,k) $
which   is $ {\cal P} \otimes {\cal B}(\R^3) 
- $ measurable, and such that
 $g(.,0,0,0) \in {\mathbb H}^2$.
 
A driver $g$ is called a $\lambda$-{\em admissible driver} if moreover there exists a constant $ C \geq 0$ such that 
$dP \otimes dt$-a.s.\,,
for each $(y_1, z_1, k_1)$, $(y_2, z_2, k_2)$,
\begin{equation}\label{lip}
|g(\omega, t, y_1, z_1, k_1) - g(\omega, t, y_2, z_2, k_2)| \leq
C ( |y_1 - y_2| +|z_1 - z_2| +   \sqrt \lambda_t |k_1 - k_2 |).
\end{equation}
%and
%\begin{equation}\label{lipo}
%(g(\omega, t, y_1, z, k) - g(\omega, t, y_2, z, k))(y_1 - y_2)\leq
%C |y_1 - y_2| ^2 .
%\end{equation}
The positive real $C$ is called the $\lambda$-{\em constant} associated with driver $g$.
\end{definition}
Note that condition \eqref{lip} implies that  for each $\,t > \vartheta$, since $\lambda_t=0$,
$g$ does not depend on $k$. 
In other terms, for each $(y,z,k)$, we have: 
$g(t,y,z,k)= g(t,y,z,0)$, $ t > \vartheta$ $dP \otimes dt$-a.s.
% 
% Moreover, condition \eqref{lipo} is satisfied if e.g. $g$ is Lipschitz with respect to $y$ uniformly with respect to $\omega, t,z,k$
%  (which corresponds to the usual case for BSDEs). It is also satisfied if $g$ is non increasing with respect to $y$,
% or if 
% $g$ is ${\cal C}^1$ in $y$  with $ \partial_y g \leq C$.
% Note also that the presence of the coefficient $\sqrt \lambda$ in the $\lambda$-{\em Lipschitz}  condition \eqref{lip} is an important point in our framework (see Remark \ref{coeff} in the Appendix).
%\begin{proposition} \label{existence} 
%Let $g$ be a $\lambda$-{\em admissible driver} and let $\xi \in {L}^2({\cal G_T})$.
%There exists an unique solution  $(X(T, \xi), Z(T, \xi), K(T, \xi))$ (denoted simply by
% $(X, Z, K)$)  in $ \mathcal{S}^2 \times {\mathbb H}^2 \times  {\mathbb H}^2_{\lambda}$ of the following BSDE:
%\begin{equation}\label{BSDE}
%-dX_t = g(t,X_t, Z_t,K_t ) dt -  Z_t dW_t - K_t dM_t; \quad
%X_T=\xi.
%\end{equation}
% \end{proposition}
Let $x \in {\mathbb R}$ be the initial wealth and let $\varphi=(\varphi^1, \varphi^2)$ in ${\mathbb H}^2 \times  {\mathbb H}^2_{\lambda}$ be a portfolio strategy.

We suppose that  the associated {\em wealth} process  $V^{x, \varphi}_t$ (or simply $V_t$)
satisfies  the following dynamics:
 \begin{equation}\label{weaun}
-dV_t= g(t,V_t, {\varphi_t}' \sigma_t , - \varphi_t^{2} ) dt - {\varphi_t}' \sigma_t dW_t +\varphi_t^{2} dM_t, 
 \end{equation}
 with $V_0=x$. Since $g$ is lipschitz with respect to $y$, this formulation makes sense. Indeed, setting $ f^1_t := \int_0^t {\varphi_t}' \sigma_t  dW_s +\varphi_t^{2} dM_s$, for each $\omega$, the deterministic function $
 (V_t^{Y_0, \varphi }(\omega))$ is defined as the unique solution of the following deterministic differential equation:
 \begin{align}\label{riun}
V_t^{x, \varphi }(\omega) = x-\int_0^t 
g(\omega, s,V_s^{x, \varphi }(\omega),{\varphi_s}' \sigma_s(\omega), - \varphi_s^{2}(\omega) )ds + f^1_t(\omega)
, \,\, 
0 \leq t \leq T.
 \end{align}
 
%We suppose that  the {\em wealth} process  $V^{x, \varphi}_t$ (or simply $V_t$)
%associated with an initial wealth $x$ and 
%a strategy $\varphi=(\varphi^1, \varphi^2)$ in ${\mathbb H}^2 \times  {\mathbb H}^2_{\lambda}$ satisfies  the following dynamics:
% \begin{equation}\label{weaun}
%-dV_t= g(t,V_t, {\varphi_t}' \sigma_t , - \varphi_t^{2} ) dt - {\varphi_t}' \sigma_t dW_t +\varphi_t^{2} dM_t, \;  V_0=x, 
% \end{equation}
Note that, equivalently, setting $Z_t= {\varphi_t}' \sigma_t$ and
  $K_t= -  \varphi_t^2 $, the dynamics \eqref{weaun} of the wealth process $V_t$ can be written as follows:
 \begin{equation}\label{wea}
-dV_t= g(t,V_t, Z_t,K_t ) dt -  Z_t dW_t - K_t dM_t.
\end{equation}
In the following, our imperfect market model is denoted by ${\cal M}^g$.\\
Note that in the case of a perfect market (see \eqref{portfolio}), we have:
\begin{equation}\label{perfectlineaire}
g(t,y,z,k) = - r_t y -  \theta^1_t z  -    \theta^2_t  k \lambda_t
%{\bf 1}_{\{t \leq \vartheta \} }
 ,
 \end{equation}
 which is a $\lambda$-admissible driver by Assumption \ref{hypo coeff}.

 % 
%such as %a borrowing interest rate $R_t $ greater than the bond rate $r$,
%  taxes or the presence of a large investor 
% \subsection{Nonlinear option pricing}
% \subsubsection{European options}
\subsection{A nonlinear pricing %and hedging of European and American options 
system}
%We now turn to the pricing and hedging problem of European options in our market model.
%Let us consider a European option with maturity $T$ and terminal payoff  $\xi \in {L}^2({\cal G_T})$ in this market model. 
 Pricing and hedging  European options in the imperfect market ${\cal M}^g$ leads to BSDEs with nonlinear driver $g$ and a default jump. By   \cite[Proposition 2.6]{DQS4}, 
we have 
\begin{proposition} \label{existence} Let  $g$ be a $\lambda$-admissible driver, let $\xi \in {L}^2({\cal G_T})$.
There exists an unique solution  $(X(T, \xi), Z(T, \xi), K(T, \xi))$ (denoted simply by
 $(X, Z, K)$)  in $ \mathcal{S}^2 \times {\mathbb H}^2 \times  {\mathbb H}^2_{\lambda}$ of the following BSDE:
\begin{equation}\label{BSDE}
-dX_t = g(t,X_t, Z_t,K_t ) dt  -  Z_t dW_t - K_t dM_t; \quad
X_T=\xi.
\end{equation}
\end{proposition}
Let us consider a European option with maturity $T$ and terminal payoff  $\xi \in {L}^2({\cal G_T})$ in this market model. Let $(X, Z, K)$ be 
the solution of BSDE \eqref{BSDE}.
The process $X$ is equal to  the wealth process associated with initial value $x= X_0$,
strategy $\varphi $ $= \Phi  ( Z,K)$ (where $\Phi$ is defined in Definition \ref{stbis})  that is
 $X= V^{X_0, \varphi}$.
Its initial value $X_0=X_0(T, \xi)$  is thus a sensible price 
%called the {\em hedging price}
 (at time $0$)  of the claim $\xi$ for the seller since this amount allows him/her to construct a trading 
strategy  $\varphi $ $\in {\mathbb H}^2 \times  {\mathbb H}^2_{\lambda}$, called {\em hedging strategy} (for the seller),  such that the value of the associated portfolio is equal to $\xi$ at time $T$. 
Moreover, by the uniqueness of the solution of BSDE \eqref{BSDE}, it is the unique price (at time $0$) which satisfies this hedging property. Similarly, $X_t=X_t(T, \xi)$ satisfies an analogous property at time $t$, and is called the {\em hedging price}
% {\em hedging price} 
 at time $t$.
 %
% of the claim $\xi$.
This  leads  to a {\em nonlinear pricing} system, first introduced by El Karoui-Quenez (\cite{EQ96})
 in a Brownian framework (later called 
 {\em $g$-evaluation} in 
\cite{Peng2004}) and denoted by ${\cal E}^g$.
For each $S\in [0,T]$, for each $\xi \in {L}^2({\cal G_S})$ 
 the associated 
$g$-evaluation is defined by 
${\cal E}_{t,S}^{^{g}} (\xi):= X_t(S, \xi)$ for each $t \in [0,S]$.

In order to ensure the (strict) monotonicity and the no arbitrage property of the nonlinear pricing system ${\cal E}^g$, we make the following assumption (see 
 \cite[Section 3.3]{DQS4}).
% We thus assume from now on that the driver $f$ satisfies the following assumption,  which 
%%and  for RBSDEs with jumps (see Theorem \ref{thmcomprbsde} below), 
%ensures  the  monotonicity property of $\rho$ by the comparison theorem for BSDEs  with jumps  (see \cite{QuenSul}, Th 4.2).
%%, and needs some additional assumption. 
%Note also that, in Insurance, the functional $-\rho=X$ can represent a risk premium. In that case, the non decreasing property with respect to financial position is often required, but also needs such an additional assumption.
%Let $T >0$ be a fixed  horizon time.
\begin{assumption}\label{Royer} 
%A driver $f$ is said to satisfy Assumption~\ref{Royer} if the following holds:\\
Assume that there exists a bounded map \begin{equation*}
 {\bf \gamma}:  [0,T]  \times \Omega\times \R^4   \rightarrow  \R \,; \, (\omega, t, y,z, k_1, k_2) \mapsto 
\gamma_t^{y,z,k_1,k_2}(\omega)
\end{equation*}
 ${\cal P } \otimes {\cal B}(\R^4) $-measurable and satisfying $ dP\otimes dt $-a.s.\,, for each $(y,z, k_1, k_2)$ $\in$ $\R^4$,
%$$g( t,y,z, k_1)- g(t,y,z, k_2) \geq  \gamma_t^{y,z, k_1,k_2} (k_1 - k_2 ) \lambda_t,$$
\begin{equation} \label{critere}
g( t,y,z, k_1)- g(t,y,z, k_2) \geq  \gamma_t^{y,z, k_1,k_2} (k_1 - k_2 )  \lambda_t,
%\quad  0 \leq t \leq \vartheta\,\,,
\end{equation} 
and $P$-a.s.\,, for each $(y,z, k_1, k_2)$ $\in$ $\R^4$,
$\gamma_{t}^{y,z, k_1, k_2} > -1$.
%0 \leq t \leq \vartheta.

%\begin{equation*}
%\text{with } \;\; \gamma:  [0,T]  \times \Omega\times \R^4   \rightarrow  \R \,; \, (\omega, t, y,z, k_1, k_2) \mapsto 
%\gamma_t^{y,z,k_1,k_2}(\omega)
%\end{equation*}
% ${\cal P } \otimes {\cal B}(\R^4) $-measurable and satisfying $ dP\otimes dt $-a.s.\,, for each $(y,z, k_1, k_2)$ $\in$ $\R^4$,
% \begin{equation}\label{condi}
%\gamma_t^{y,z, k_1, k_2} > -1 \,\,\,  \;\; \text{ and }
%\,\,  \;\;|\gamma_t^{y,z, k_1, k_2}|  \leq c', \quad  0 \leq t \leq \vartheta\,\,,
%\end{equation}
%where $c'$ is a positive constant.
%$C_1$ where $C_1 > -1$.
\end{assumption}
\noindent 
%Recall that $\lambda$ vanishes after $\vartheta$ and $g(t,\cdot)$ does not depend on $k$ 
%on $\{t >\vartheta\}$. Hence, the inequality \eqref{critere} is always satisfied on $\{t >\vartheta\}$.\\
This assumption is satisfied e.g. when $g(t,\cdot)$ is non decreasing with respect to $k$, 
%on $\{t \leq \vartheta\}$, 
or if 
 $g$ is ${\cal C}^1$ in $k$  with $ \partial_k g(t, \cdot) > - \lambda_t$ on $\{t \leq \vartheta\}$. In the special case of a perfect
  market, $g$ is given by \eqref{perfectlineaire}, which implies  that $ \partial_k g(t, \cdot)=- \theta^2_t\lambda_t$. In this case,
 Assumption \ref{Royer}  is thus equivalent to $ \theta^2_t<1$, which corresponds to the usual assumption \eqref{cth} made in the literature on default risk.
 
\begin{remark}\label{prixnul}
Suppose that $g(t,0,0,0)=0$ $dP\otimes dt $-a.s.\, 
Then the price of an option with a null payoff is equal to $0$, that is, for each $S\in [0,T]$, 
${\cal E}^{^{g}}_{\cdot, S} (0)= 0$ a.s.
%Moreover, the price system ${\cal E}^{^{g}}$ then satisfies the {\em zero-one law.
Moreover,  by the comparison theorem for BSDEs with default jump  (see \cite[Theorem 2.17]{DQS4}), it follows that 
the nonlinear pricing system  ${\cal E}^{^{g}}$ is nonnegative, that is, for each $S\in [0,T]$, for all $\xi \in {L}^2({\cal G_S})$,
if $\xi \geq 0$ a.s., then ${\cal E}^{^{g}}_{\cdot, S} (\xi)\geq 0$ a.s.
\end{remark}

%\begin{remark}\label{positive}
%%
%{\bf
% Suppose that Assumption \ref{Royer} holds and that $g(t,0,0,0) \geq 0$ $ dP\otimes dt $-a.s. By the comparison theorem for BSDEs with default jump  (see \cite[Theorem 2.17]{DQS4}), we derive that the nonlinear pricing system  ${\cal E}^{^{g,\cdot}}$ is nonnegative, that is, for each $S\in [0,T]$, for all $\xi \in {L}^2({\cal G_S})$,
%if $\xi \geq 0$ a.s., then ${\cal E}^{^{g}}_{\cdot, S} (\xi)\geq 0$ a.s.
%}
%%
%\end{remark}
% This assertion can be shown by classical analysis arguments,  similar to those used in the proof of Proposition A.2
% in \cite{DQS2}.\\
% By  the comparison theorems for BSDEs (see Proposition ... in \cite{BSDEdefault}, see also Section..), 
 \begin{definition}\label{defmart}
Let $Y \in S^2$. The process $(Y_t)$ is said to be a strong ${\cal E}$-supermartingale (resp. martingale)  if ${\cal E}_{\sigma ,\tau}(Y_{\tau}) \leq Y_{\sigma}$ (resp. $= Y_{\sigma}$) a.s. on $\sigma \leq \tau$,  for all $ \sigma, \tau \in \mathcal{T}_0$. 
\end{definition}

\begin{proposition}  \label{rima}
For each $S\in [0,T]$ and for each $\xi \in {L}^2({\cal G_S})$, the associated price (or $g$-evaluation)
${\cal E}_{t,S}^g (\xi)$ is an $\mathcal{E}^g$-martingale.
Moreover, for each $x \in \mathbb{R}$ and each portfolio strategy $\varphi$ $\in$ ${\mathbb H}^2\times 
{\mathbb H}^2_{\lambda}$, the associated wealth process $V^{x, \varphi}$ is an $\mathcal{E}^g$-martingale.
\end{proposition}

\begin{proof} 
By the flow property of BSDEs, the solution of a BSDE with driver $g$ is an 
$\mathcal{E}^g$-martingale. The first assertion 
%directly
 follows. The second one is obtained by noting that 
%$V_t^{x, \varphi}= {\cal E}_{t,T}^g (V_T^{x, \varphi})$, $0\leq t \leq T$, since 
$V^{x, \varphi}$ is the solution of the 
BSDE with driver $g$, terminal time $T$ and terminal condition $V_T^{x, \varphi}$.
% JUSQU'ICI
\end{proof}

\begin{example}[Examples of market imperfections] \label{eximp} 

\
%\em
\begin{itemize}
\item Different borrowing and lending interest rates $R_t$ and $r_t$, with $R_t\geq r_t$:  the driver $g$ is then of the  form %SUPPRIMER LES INDICATRICES?
\begin{equation*} \label{largeinvestor}
g(t, V_t, \varphi_t '\sigma_t ,  - \varphi_t^{2}  )=- r_t V_t-\varphi_t^1 (\mu^1_t - r_t)- \varphi_t^2(\mu^2_t - r_t)   +  (R_t-r_t) (V_t -\varphi_t^1 - \varphi_t^2 )^-,
\end{equation*} 
where $\varphi_t^2$ vanishes after $\vartheta$ and 
%This models {\em different borrowing and lending interest rates}, denoted respectively by $R_t$ and $r_t$ 
 (see e.g. \cite{CviK}).

\item  Large investor seller:
Suppose  that the seller of the option  is a  large trader whose hedging strategy $\varphi$ and its associated cost $V$ may influence the market prices
%Moreover, the price of the option (which often corresponds to the cost of
%  the hedging portfolio for the large investor seller) can also influence the market prices (cf. \cite{G}). 
  (see e.g. \cite{CM, BK2}). 
 Taking into account  the possible feedback effects in the market model, the large trader-seller  may suppose that the coefficients 
%$r_t, \sigma_t^1$, $\sigma_t^2$ and $\beta_t$ 
%depend on the  value of the portfolio $V_t$ and on  the strategy $\varphi_t$. 
%The coefficients can thus 
are of the form $\sigma_t (\omega)= \bar \sigma (\omega, t, V_t, \varphi_t)$ where $
\bar \sigma:   \Omega \times [0,T]  \times \mathbb{R}^3 \,; \, ( \omega,t, x, z,k) \mapsto 
 \bar \sigma ( \omega, t, x, z, k)
$ is a ${\cal P } \otimes {\cal B}({\mathbb R}^3) $-measurable map,
% , uniformly Lipschitz 
% with respect to $x,z,k$ and bounded, 
 and similarly for the other coefficients $ r$, $ \mu^1$, $ \mu^2$. The driver is thus of the form:
 $$g(t, V_t, \varphi '_t \bar \sigma_t (t,  V_t, \varphi_t), -\varphi_t^{2}  )=- \bar r (t,  V_t, \varphi_t) V_t-\varphi_t^1 \,
 (\bar \mu^1_t - \bar r_t) (t,  V_t, \varphi_t)-\varphi_t^2(\bar \mu^2_t - \bar r_t)   (t,  V_t, \varphi_t). $$ 
% We have to assume 
% here that   
%
  Here, the map  $\Psi:$ $(\omega, t,y,\varphi) \mapsto (z,k)$ with $z={\varphi}' \bar \sigma_t(\omega,t,y,\varphi)$ and
 $k=- \varphi^2$ is assumed to be one to one with respect to $\varphi$, and such that its inverse $\Psi^{-1}_{\varphi} $ is ${\cal P}\otimes {\cal B} ({\bf R}^3)$-measurable.

\item  Taxes on risky investments profits:  Let  $\rho$ $\in$ $]0,1[$ represents an instantaneous tax coefficient (see e.g. \cite{EPQ01}). The driver is then given by:
$$g(t, V_t, \varphi '_t \sigma_t , \varphi_t^{2}\beta_t  )=- r_t V_t-\varphi_t^1 (\mu^1_t - r_t)- \varphi_t^2(\mu^2_t - r_t)  + \rho (\varphi^1_t + \varphi^2_t )^+.$$
\end{itemize}
\end{example}

% \begin{remark}A GARDER?
%Note that  we can also suppose that the coefficients $\beta$ and $\gamma$ depend also on the superhedging strategy 
%$\varphi = \Phi (Z,K)$. This makes sense when the seller is a large trader, whose hedging strategy can affect the default intensity. It is then important for him to take into account these feedback effects when he defines his model. 
%Note that in this case, the a priori probability measures will depend not only on the ambiguity coefficient $\alpha$ but also on $(Z,K)$.
%\end{remark}

\section{Pricing and hedging of game options in the imperfect market ${\cal M}^g$}\label{sec3}

Let $T>0$ be the terminal time.
Let $\xi$  and $\zeta$ be  adapted RCLL processes in ${S}^2$ with $\zeta_T= \xi_T$ a.s.\,and $\xi_t \leq \zeta_t$, $0 \leq t \leq T$ a.s.\,
 We suppose  that Mokobodzki's condition is  satisfied, that is  there exist two nonnegative RCLL supermartingales $H$ and $H'$ in $ \mathcal{S}^2$ such that:
\begin{equation*}\label{Moki}
\xi_t  \leq H_t -H'_t \leq \zeta_t  \quad 0\leq t \leq T \quad {\rm a.s.}
\end{equation*}

%A game contingent claim is a contract between a seller and a buyer which enables the seller to cancel it and the buyer to exercise it at any time before the maturity $T$. If the buyer exercises the contract at time $t$, then the seller pays to him the amount $\xi_{t}$, but if the seller cancels  before the buyer exercises , then he pays to the buyer the amount $\zeta_{t}$ at the cancel time. If the buyer exercises at the same time $t$ as the seller cancels, then the seller pays to him the amount $\xi_t$. The difference $\zeta_{t} - \xi_t \geq 0$ is interpreted as a penalty that the seller pays to the buyer for the cancellation of the contract. 

The game option consists for the seller to select a cancellation time $\sigma \in {\cal T}$ and for the buyer to choose an exercise time $\tau \in {\cal T}$, so that the seller pays to the buyer at time $\tau \wedge \sigma$
the amount 
\begin{equation*} \label{terminal}
I(\tau,\sigma):=\xi_{\tau} {\bf 1} _{\tau \leq \sigma}+ \zeta_{\sigma} {\bf 1} _{\sigma < \tau}.
\end{equation*}
  
 We now introduce the {\em seller's price} of the game option, denoted by $u_0$, defined as the infimum of the initial wealths which enable the seller to choose 
a cancellation time $\sigma$  and to 
construct a portfolio which will cover  his liability to pay the payoff to the buyer up to $\sigma$ no matter the exercise time chosen by the buyer.
% Moreover, it is sometimes used as a synonymous of the notion of {\em fair value accounting},
% introduced in 2006 by the Financial Accounting Standards Board.
 % RATIONALIT� �CONOMIQUE ET JUSTE PRIX
%Arnaud Berthoud
%Cahiers d'�conomie politique / Papers in Political Economy
%No. 19, LE MARCH� CHEZ ADAM SMITH (1991), pp. 139-156
% For example, a buyer may consider that to be fair, a price must be in line with the fair market value of the contract deliverable.
% However, even in the case of a perfect market model, this terminology is not 
%really appropriate since in finance, 
%the term Fair Price is defined as a reference price, estimated by statistics methods, paid by comparable
% customers. 
% The $g$-value of the game option is To be fair to the supplier a price must be realistic in terms of the supplier's ability to satisfy the terms and conditions of the contract.
% Here, the superhedging price can be interpreted as a fair (theoretical) price from the seller's
%  point of view.
\begin{definition}
For each initial wealth $x$, a {\em super-hedge} against the game option is a pair $(\sigma, \varphi)$ of a stopping time $\sigma \in {\cal T}$ and a portfolio strategy $ \varphi$ $\in$  ${\mathbb H}^2 \times  {\mathbb H}^2_{\lambda}$ such that \footnote{Note that condition \eqref{condA} is equivalent to 
%\begin{equation}\label{I}
$V^{x, \varphi}_{t \wedge \sigma} \geq I(t,\sigma), \;\; 0 \leq t \leq T \quad {\rm a.s.}$}
\begin{equation}\label{condA}
V^{x, \varphi}_{t } \geq \xi_t,  \; 0\leq t \leq \sigma \; \text{ a.s. and } V^{x, \varphi}_{\sigma } \geq \zeta_{\sigma}
 \text{ a.s.} 
\end{equation}
We denote by ${\cal S} (x)= {\cal S}_{\xi, \zeta}(x) $ the  set of all super-hedges associated with initial wealth $x$.

 We define the {\em seller's price} as
\begin{equation} \label{seller's price}
 u_0:= \inf \{x \in \R,\,\, \exists  (\sigma, \varphi) \in {\cal S} (x) \}.
 \end{equation}
 When the infimum in \eqref{seller's price} is attained,  
the amount $u_0$ allows the seller to be super-hedged, and is called  the {\em superhedging price}. 
\end{definition}
%\begin{remark}
%\end{equation} 
%\end{remark}
%When the infimum is not attained,  the amount $u_0$ does not necessarily allow the seller to be super-hedged. In this case, the seller's price 
%$u_0$ is called 
% {\em nearly superhedging price}.
%% However,  $u_0$ is still called the  superhedging price
%%of the game option (as in the  financial literature related to superhedging prices).
%By definition of $u_0$ as an infimum, for each $\varepsilon>0$, there exists a super-hedge associated with initial wealth $u_0 + \varepsilon$. In other terms, the amount $u_0$ is nearly sufficient 
%to  allow the seller to be super-hedged. 
\begin{remark}\label{positive2}
We have $(0,0) \in {\cal S} (\zeta_0)$ since $V^{\zeta_0,0}_0= \zeta_0$ and $\zeta_0\geq \xi_0$. By \eqref{seller's price}, we thus get $u_0 \leq \zeta_0$.

Moreover, when $g(t,0,0,0)=0$ $dP\otimes dt $-a.s.\, and $\zeta \geq 0$, then we can restrict ourselves to nonnegative initial wealths, that is   $u_0= \inf \{x \geq 0,\,\, \exists  (\sigma, \varphi) \in {\cal S} (x) \}$.
Indeed, let $x \in \R$ be such that there exists $(\sigma, \varphi) \in {\cal S} (x) $. Then, $V^{x, \varphi}_{\sigma } \geq \zeta_{\sigma} \geq 0$ a.s. Now, by Proposition \ref{rima}
the wealth process $V^{x, \varphi}$ is an $\mathcal{E}^g$-martingale. We thus have $x= \mathcal{E}^g_{0, \sigma}(V^{x, \varphi}_{\sigma })$. Since the pricing system  $\mathcal{E}^g$ is nonnegative (see Remark \ref{prixnul}), it follows that $x= \mathcal{E}^g_{0, \sigma}(V^{x, \varphi}_{\sigma })\geq 0$. 
\end{remark}

We now provide  a dual formulation of the seller's price, expressed in terms of 
the nonlinear pricing system ${\cal E }^g$.%
%characterize the seller's price as the value function of  a {\em generalized Dynkin game} recently introduced in \cite{DQS2}. 
We introduce the following definition:
 \begin{definition}
 We define the {\em $g$-value} of the game option as 
 \begin{equation} \label{prixg}
%Y(0): =  
\inf_{\sigma \in \mathcal{T} }  \sup_{\tau \in \mathcal{T}} \cal{E}^g_{0,\tau \wedge \sigma}[I(\tau, \sigma)].
\end{equation}
\end{definition}
%By a result  of  \cite[Theorem 4.9]{DQS2},
% the above {\em generalized Dynkin game} problem is {\em fair}, that is
%  $ \inf_{\sigma  }  \sup_{\tau }= \sup_{\tau } \inf_{\sigma  }$.
%and will be called in the sequel.  

%Using the characterization of the {\em $g$-value}
%%$\inf_{\sigma  }  \sup_{\tau } \cal{E}_{0,\tau \wedge \sigma}[I(\tau, \sigma)]$
%of the game option
% as  the solution of the doubly reflected  
%BSDE \eqref{DRBSDE} (see Proposition \ref{fairpricegame}), 
Our aim is to  show that the seller's price $u_0$ of the game option is equal to its $g$-value. To this purpose, we first give the following characterization of the  $g$-value.

 %tuse the following characterization of the common value of the {\em generalized Dynkin game} as the solution  of a doubly reflected BSDE.
\begin{proposition}\label{fairpricegame} (Characterization of the {\em $g$-value} of the game option)
Suppose that the payoffs $\xi$ and $\zeta$ are (only) RCLL. 
The {\em $g$-value} of the game option satisfies:
%defined by \eqref{prixg}, is equal to the common value of 
%the {\em generalized Dynkin game} associated with the criterium $\cal{E}^g_{0,\tau \wedge \sigma}[I(\tau, \sigma)]$, $(\tau, \sigma) \in \mathcal{T}^2$, that is
\begin{equation}\label{Dynkin}
% Y(0)=
 \inf_{\sigma \in \mathcal{T} }  \sup_{\tau \in \mathcal{T}} \cal{E}^g_{0,\tau \wedge \sigma}[I(\tau, \sigma)]= 
\sup_{\tau\in \mathcal{T} } \inf_{\sigma \in \mathcal{T}} 
\cal{E}^g_{0,\tau \wedge \sigma}[I(\tau, \sigma)]=Y_0,
\end{equation}
where $(Y, Z, K, A,A')$ is the unique solution in  ${S}^2 \times {\mathbb H}^2 \times {\mathbb H}^2_{\lambda}\times {\cal A}^2 \times {\cal A}^2$ of the doubly  reflected BSDE (DRBSDE)  
associated with driver $g$ and barriers $\xi, \zeta$, that is
 \begin{align}
   -dY_t &= g(t,Y_t,  Z_t, K_t )dt +dA_t -dA_t^{'}- Z_t  dW_t -K_t dM_t; \;  Y_T = \xi_T, \label{DRBSDE} \\
   \text{with} &  \nonumber \\
  &(i)   \;\; \xi_t \leq Y_t \leq \zeta_t,  \; 0 \leq t \leq T \text{ a.s.}, \nonumber \\
 &(ii)   \; \; dA_t \perp dA_t'  \quad \text {(i.e. the measures $dA_t$ and $dA_t'$ are mutually singular)}\nonumber \\
& (iii)   \displaystyle   \int_0^T (Y_t - \xi_t) dA^c_t = 0 \text{ a.s. and } \; \displaystyle    \int_0^T (\zeta_t-Y_t)dA^{'c}_t = 0 \text{ a.s.  } \nonumber \\
& \qquad  \Delta A_{\tau}^d=  \Delta A_{\tau}^d {\bf 1}_{\{Y_{\tau^-} = \xi_{\tau^-}\}} \text{ and }  \;  \Delta A_{\tau}^{'d}= \Delta A_{\tau}^{'d} {\bf 1}_{\{Y_{\tau^-} = \zeta_{\tau^-}\}} \text{ a.s. } \forall \tau \in {\cal T} \text{ predictable. } \nonumber
\end{align}
\end{proposition}

%\begin{remark}\label{sauts} {\bf NOUVEAU
%Note that for each predictable stopping time $\tau$, by \eqref{DRBSDE}, we have 
%$(\Delta Y_{\tau})^+ = \Delta A'_{\tau}$ a.s. and $(\Delta Y_{\tau})^- = \Delta A_{\tau}$ a.s. 
%}
%\end{remark}

 %The notation $dA_t \perp dA_t'$ means that . \\
%, under  a left-regularity assumption on $\zeta$ (but not on $\xi$), the superhedging 
%price of the game option is well defined, and is equal to its $g$-value.
%
Using the terminology introduced in \cite{DQS2}, the first equality in \eqref{Dynkin}
means that the {\em generalized Dynkin game} associated with the criterium $\cal{E}^g_{0,\tau \wedge \sigma}[I(\tau, \sigma)]$ is {\em fair}.

When $g$ is linear and when there is no default, this  
corresponds to a well-known result on classical Dynkin games and linear DRBSDEs (see e.g. \cite{CK,H}). 

\begin{proof} The proof of existence and uniqueness of a solution  $(Y, Z, K, A,A')$ of the DRBSDE \eqref{DRBSDE} is given in  appendix. 
Proceeding as in the proof of  \cite[Theorem 4.9]{DQS2} which was given in the framework of a random Poisson measure, we can prove that for each $S \in {\cal T}$,
$
Y_S= {\rm ess} \inf_{\sigma \in {\cal T}_S }\,{\rm ess} \sup_{\tau \in {\cal T}_S} \, \cal{E}^g_{S,\tau \wedge \sigma}[I(\tau, \sigma)]=
{\rm ess} \sup_{\tau \in {\cal T}_S} \, {\rm ess} \inf_{\sigma \in {\cal T}_S} \, \cal{E}^g_{S,\tau \wedge \sigma}[I(\tau, \sigma)]$ a.s.
%the {\em generalized Dynkin game} is fair and that its common value is equal to $Y_0$. 
The results of the proposition then follow by taking $S=0$. 
\end{proof}

\begin{proposition}\label{lusc}
Let $(Y, Z, K, A,A')$ be the unique solution of the DRBSDE \eqref{DRBSDE}. 
 When $\xi$ (resp. $- \zeta$) is left-u.s.c. along stopping times, then 
$A$ (resp. $A'$) is continuous.
\end{proposition}
\begin{proof} 
Note first that for each predictable stopping time $\tau$, by \eqref{DRBSDE}, we have 
$(\Delta Y_{\tau})^+ = \Delta A'_{\tau}$ a.s. and $(\Delta Y_{\tau})^- = \Delta A_{\tau}$ a.s. 
Suppose that now that $- \zeta$ is left-u.s.c. along stopping time. Let $\tau$ be a 
predictable stopping time.
 Using the equality
$\Delta A'_{\tau}=(\Delta Y_{\tau})^+$ together with the Skorokhod conditions satisfied by 
$A'$, we get 
\begin{equation} \label{abcd}
\Delta A'_{\tau}= {\bf 1}_{ \{Y_{\tau^-} = \zeta_{\tau^-}\} }( Y_{\tau}- Y_{\tau^-})^+= 
{\bf 1}_{ \{Y_{\tau^-} = \zeta_{\tau^-}\} }( Y_{\tau}- \zeta_{\tau^-})^+.
\end{equation}
Now, since $- \zeta$ is left-u.s.c. along stopping times, we have 
$Y_{\tau}- \zeta_{\tau^-} \leq Y_{\tau}- \zeta_{\tau} \leq 0$ a.s.\,, where the last equality follows 
from the inequality $Y \leq \zeta$. Using \eqref{abcd}, we derive that $\Delta A'_{\tau}= 0$ a.s.
It follows that 
$A'$ is continuous. By similar arguments, one can show that if $\xi$ is left-u.s.c. along stopping times, then 
$A$ is continuous.
%The proof is similar to the proof of    \cite[Theorem 3.7(i)]{DQS2}. 
\end{proof}

%Using the characterization of the {\em $g$-value}
%%$\inf_{\sigma  }  \sup_{\tau } \cal{E}_{0,\tau \wedge \sigma}[I(\tau, \sigma)]$
%of the game option
% as  the solution of the doubly reflected  
%BSDE \eqref{DRBSDE} (see Proposition \ref{fairpricegame}), 
Using the above propositions, we can now show the dual formulation for the seller's price. We first consider the simpler case when $\zeta$ is left lower-semicontinuous (or equivalently $-\zeta$ is left-u.s.c.) along stopping times. In this case, we prove below that the seller's price is equal to the $g$-value 
and that  the infimum in \eqref{seller's price} is attained. This implies that the seller's price is the {\it super-hedging} price. Moreover, a super-hedge strategy is provided via the solution of the associated DRBSDE.

\begin{theorem}[Seller's/super-hedging price and super-hedge of the game option] \label{superhedging} 

Suppose that  $\zeta$ is left lower-semicontinuous along stopping times (and $\xi$ is only RCLL). 
The seller's price  \eqref{seller's price} of the game option coincides with the {\em $g$-value} of the game option, that is 
\begin{equation}\label{u0}
u_0=\inf_{\sigma \in \mathcal{T} }  \sup_{\tau \in \mathcal{T}} \cal{E}^g_{0,\tau \wedge \sigma}[I(\tau, \sigma)]= 
\sup_{\tau\in \mathcal{T} } \inf_{\sigma \in \mathcal{T}} 
\cal{E}^g_{0,\tau \wedge \sigma}[I(\tau, \sigma)].
\end{equation}
Let $(Y, Z, K, A,A')$ is the  solution of the 
DRBSDE  
associated with driver $g$ and barriers $\xi, \zeta$. The seller's price is equal to $Y_0$, 
  that is 
%Assume that $A'$ are continuous, which is satisfied if for example $- \zeta$ is left u.s.c. along stopping times.
 $$u_0=Y_0.$$
Moreover, the infimum in \eqref{seller's price} is attained. The seller's price is thus the {\em super-hedging price} and 
there exists a super-hedge strategy $(\sigma^*, \varphi^*)$ associated with the initial amount $u_0$, given by
\begin{equation}\label{sigmaetoile}
\sigma^*:= \inf \{ t \geq 0,\,\, Y_t = \zeta_t\} \quad {\rm and} \quad \varphi^*:=\Phi(Z,K),
\end{equation}
 where $\Phi$ is defined in Definition \ref{stbis}.
% JUSQU'ICI
% The pair $(\sigma^*, \varphi^*)$  belongs to ${\cal S}(u_0)$. In other terms, the cancellation time $\sigma^*$ and the risky assets strategy $\varphi^*$  allow the seller of the game option to be super-hedged. 
% Moreover, the pair $(\bar \sigma, \varphi^*)$, where $\bar \sigma:= \inf \{ t \geq 0,\,\, A_t >0\}$ also 
%belongs to ${\cal S}(u_0)$.
%\begin{align}
%u_0=Y_0=\sup_{\tau \in \mathcal{T}} \inf_{\sigma \in \mathcal{T}} \mathcal{E}_{0,T}^f[I(\tau, \sigma)]=\inf_{\sigma \in \mathcal{T}} \sup_{\tau \in \mathcal{T}} \mathcal{E}_{0,T}^f[I(\tau, \sigma)].
%\end{align}
\end{theorem}

\begin{remark}
In the  special case of a perfect market model, our result gives that $u_0$ is characterized as 
% where the driver $g$ is linear with respect to  $y,z,k$, 
 the value function of a classical Dynkin game problem, which is shown in the literature (see e.g. \cite{Kifer,H}) under an additional regularity assumption on $\xi$, by using an actualization procedure, a change of probability
 measure, and some results on classical Dynkin games. Moreover, in this particular case, the characterization of $u_0$ and of the super-hedge via the solution of a {\em linear} doubly reflected BSDE  are shown in 
\cite{H} by using  the links  between {\em linear} DRBSDEs and classical Dynkin games (first provided in \cite{CK}). 
To solve the  problem in the case of an imperfect market model, when $g$ is nonlinear, 
we need to use other arguments, in particular some properties of the nonlinear $g$-evaluation ${\cal E}^g$, comparison theorems  for backward SDEs and for forward differential equations, and the links between {\em nonlinear} doubly reflected BSDEs and {\em generalized Dynkin
 games} (first provided in \cite{DQS2}).
  \end{remark}

\begin{proof}  
By Proposition \ref{fairpricegame}, the {\em $g$-value} of the game option is equal $Y_0$. 
Note that $u_0= \inf \mathcal{H}$,
%
%$\sup_{\nu \in {\mathcal T}}{\cal E}_{0,\nu}^g ( \xi_{\nu})$
where $\mathcal{H}$  is the set of initial capitals which allow the seller to be {\em super-hedged}, that is
$$\mathcal{H}= \{ x \in \mathbb{R}: \exists ( \sigma,\varphi) \in {\cal S}(x) \}. $$
Let us show that $Y_0 \geq u_0$. It is sufficient to prove that there exists 
$( \sigma^*,\varphi^*) \in {\cal S}(Y_0)$.
By Proposition \ref{lusc}, since $- \zeta$ is left-u.s.c. along stopping times, the process  
$A'$ is continuous.
%Let  $\sigma^*:= \inf \{ t \geq 0,\,\, Y_t = \zeta_t\}.$
Let $\sigma^*$ be defined as in \eqref{sigmaetoile}.  We have a.s. that $Y_t< \zeta_t$ for each $t \in$ $[0, \sigma^*[$.  Since $Y$ is solution of the DRBSDE \eqref{DRBSDE}, the process $A'$ is thus constant on $[0, \sigma^*[$ a.s.\, and even on $[0, \sigma^*]$ by continuity. Hence, $A'_{\sigma^*}=A'_0=0$ a.s. 
%Since $A'$ is continuous, 
%the non decreasing processes are continuous, there exists a saddle point $(\sigma^*, \tau^*) \in \mathcal{T} \times \mathcal{T}$ for the generalized Dynkin game.
 For almost every $\omega$, we thus have
\begin{align}\label{forwardb}
Y_t(\omega)=Y_0-\int_0^t g(s,\omega, Y_s(\omega),Z_s(\omega),K_s(\omega))ds+f_t(\omega)-A_t(\omega), \,\, 
0 \leq t \leq \sigma^*(\omega).
\end{align}
where $f_t:= \int_0^t Z_sdW_s+
\int_0^tK_sdM_s$.
Now, the wealth $V_.^{Y_0, \varphi ^*}$, associated with the initial capital $Y_0$ and the financial strategy 
$\varphi^*:=\Phi(Z,K)$ satisfies for almost every $\omega$ the forward deterministic differential equation:
\begin{align}\label{ri}
V_t^{Y_0, \varphi ^*}(\omega) = Y_0-\int_0^t g(s,V_s^{Y_0, \varphi ^*}(\omega),Z_s(\omega),K_s(\omega))ds + f_t(\omega)
, \,\, 
0 \leq t \leq T.
 \end{align}
%Note that the coefficient $(s,x) \mapsto -g(s,\omega, x,Z_s(\omega),K_s(\omega))$  is Lipschitz  continuous with respect to $x$.
% for almost every $\omega$ the equation 
%\begin{align}\label{forwardc}
%V_t^{Y_0, \varphi ^*}(\omega) = Y_0-\int_0^t g(s,V_s^{Y_0, \varphi ^*}(\omega),Z_s(\omega),K_s(\omega))ds + f_t(\omega)
%, \,\, 
%0 \leq t \leq T.
% \end{align}
%where  $b(s, \omega, x):= -g(s,\omega, x,Z_s(\omega),K_s(\omega))$ (which is Lipschitz  continuous with respect to $x$).
%Now, since $Y$ satisfies \eqref{forwardb}, we have for almost every $\omega$
%%For each $\varepsilon>0$, we introduce the stopping time:
%%\begin{align}
%%\tau_{\varepsilon}:= \inf \{t \geq 0: \,\,\, Y_t \leq \xi_t+\varepsilon \}.
%%\end{align}
%%By Lemma 3.3 (ii) in $\cite{QS2}$, the process $A$ is constant on $[\![ 0, \tau_\varepsilon ]\!].$ The dynamics of $Y$ is thus the following on this interval:
%\begin{align}\label{DM}
%Y_t(\omega)=Y_0+\int_0^tb( s, \omega, Y_s(\omega))ds+f_t(\omega)- A_t(\omega),\,\,0 \leq t \leq T.
%\end{align}
%Hence,  $Y$  coincides with the (unique) solution of the {\em forward} SDE associated with the initial condition 
%$Y_0$, the ``generalized drift" $ -g(s,\omega, x,Z_s,K_s)ds - dA_s(\omega)$, and the same ``volatility" coefficients 
%$Z_s(\omega)$, $K_s(\omega)$ as those corresponding to $V^{Y_0, \varphi ^*}$.
%
Since $A$ is non decreasing, by applying the classical comparison result on $[0, \sigma^*(\omega)]$ (see e.g.  Lemma \ref{classique}) for the two forward differential equations \eqref{forwardb} and \eqref{ri}, with the same  coefficient $(s,x) \mapsto -g(s,\omega, x,Z_s(\omega),K_s(\omega))$, we get
%The two above equations together with a comparison result of classical analysis for {\em forward} differential equations (see the proof of Proposition \ref{americano} for details) imply that 
\begin{equation*}%\label{oo}
V_t^{Y_0, \varphi ^*} \geq Y_t  \geq \xi_t, \,\,0 \leq t \leq \sigma^* \quad {\rm a.s.}\,,
\end{equation*}
 where the last inequality follows from the inequality $Y \geq \xi$. We also have
\begin{equation*}%\label{ooo}
  V_{\sigma^*}^{Y_0, \varphi ^*} \geq Y_{\sigma^*} = \zeta_{\sigma^*} \quad {\rm a.s.},
  \end{equation*}
 where the last equality follows from
   the definition of the stopping time $\sigma^*$  and the right-continuity of $Y$ and $\zeta$. Hence, 
   \begin{equation}\label{ay}
  ( \sigma^*,\varphi^*) \in {\cal S}(Y_0), 
  \end{equation}
%  By similar arguments, 
%  $( \bar \sigma,\varphi^*) \in {\cal S}(Y_0)$.
which implies that $Y_0 \in \mathcal{H}$. We thus get the inequality $Y_0 \geq u_0$. 
  
It remains to show that $u_0\geq Y_0$.
%Since $Y_0$ coincides with the value function of the generalized Dynkin game, we have $Y_0=  \sup_{\tau \in \mathcal{T}}  \inf_{\sigma \in \mathcal{T}} \mathcal{E}_{0,T}^f[I(\tau, \sigma)]=\underline{V}(0).$
%
%Since moreover $\underline{V}(0) \leq \overline{V}(0)$, in order to show the converse inequality, i.e. 
Since $Y_0= \inf_{\sigma \in \mathcal{T}} \sup_{\tau \in \mathcal{T}} \mathcal{E}_{0,T}^g[I(\tau, \sigma)]$ (by Proposition \ref{fairpricegame}),  it is sufficient to show that
\begin{equation}\label{dd}
u_0 \geq  \inf_{\sigma \in \mathcal{T}} \sup_{\tau \in \mathcal{T}} \mathcal{E}_{0,T}^g[I(\tau, \sigma)].
\end{equation}
Let $x \in \mathcal{H}$. There exists $( \sigma,\varphi) \in {\cal S}(x)$, that is 
  a pair $(\sigma, \varphi)$ of a stopping time $\sigma \in {\cal T}$ and a portfolio strategy $ \varphi$ $\in$  ${\mathbb H}^2 \times  {\mathbb H}^2_{\lambda}$ such that 
  $V^{x, \varphi}_{t } \geq \xi_t$, $0\leq t \leq \sigma$ a.s. and 
  $V^{x, \varphi}_{\sigma } \geq \zeta_{\sigma}$ a.s.\,, which implies that for 
%$V^{x, \varphi}_{t \wedge \sigma} \geq I(t,\sigma)$, $0 \leq t \leq T$ a.s. 
all $\tau \in \mathcal{T}$ we have
$$V^{x, \varphi}_{\tau \wedge \sigma} \geq I(\tau, \sigma) \quad {\rm a.s.}$$
%$$V_\sigma^{x, \varphi} \geq \xi_\tau \textbf{1}_{\{\tau < \sigma\}} + \zeta_\sigma \textbf{1}_{\{\tau=\sigma\}}$$
By taking the $\mathcal{E}^g$-evaluation in the above inequality and then the supremum on 
$\tau \in \mathcal{T}$, using the monotonicity of the $\mathcal{E}^g$-evaluation and the
$\mathcal{E}^g$-martingale property of the wealth process $V^{x, \varphi}$ (see Proposition \ref{rima}),
we obtain $x= \mathcal{E}^g _{0,\tau \wedge \sigma}[V^{x, \varphi}_{\tau \wedge \sigma}] \geq 
\mathcal{E}_{0, \tau \wedge \sigma}^g [ I(\tau, \sigma)]$, for each $\tau \in \mathcal{T}$. By taking 
the supremum over $\tau \in \mathcal{T}$, and then the infimum over $\sigma \in \mathcal{T}$, we get%\begin{align*}
%x \geq \sup_{\tau \in \mathcal{T}}
%\mathcal{E}_{0, \tau \wedge \sigma}^g [ I(\tau, \sigma)].
%\end{align*}
%By taking the infimum over $\sigma \in \mathcal{T}$, we get
\begin{align*}
x  \geq 
\inf_{\sigma \in \mathcal{T}} \sup_{\tau \in \mathcal{T}}
\mathcal{E}_{0, \tau \wedge \sigma}^g [ I(\tau, \sigma)].
\end{align*}
This inequality holds for any $x \in \mathcal{H}$. By taking the infimum over $x \in \mathcal{H}$, 
we obtain the inequality \eqref{dd}, which yields that $u_0\geq Y_0$. Since $Y_0 \geq u_0$, we get $Y_0=u_0$. Moreover, this equality together with 
\eqref{ay} implies that $( \sigma^*,\varphi^*) \in {\cal S}(u_0).$
The proof is thus complete. 
\end{proof}

%\begin{remark}\label{positive2}
%{\bf 
%Suppose that  $\xi$ and $\zeta$ are nonnegative. Then, the equality \eqref{u0} together the nonnegativity of the pricing system $\mathcal{E}^g$ yield that $u_0 \geq 0$. We thus have $u_0= \inf \{x \in \R_+,\,\, \exists  (\sigma, \varphi) \in {\cal S} (x) \}$.
%}
%\end{remark}

%\begin{remark}\label{positive2}
%{\bf 
%Suppose that  $g(t,0,0,0) \geq 0$ and that  $\xi$ and $\zeta$ are nonnegative. Then, the equality \eqref{u0} together with Remark \ref{positive} yield that $u_0 \geq 0$. We thus have $u_0= \inf \{x \in \R_+,\,\, \exists  (\sigma, \varphi) \in {\cal S} (x) \}$.
%}
%\end{remark}

\begin{remark} \label{simi}
%{\bf
% % %
%% ICI IMPORTANT: NE PAS EFFACER:
%Similarly, for each time $t\in [0,T]$, $Y_t$ is equal to the seller's price at time $t$ of the game option, that is, the (essential) infimum of the initial wealths (at initial time $t$) which enable the seller to be super-hedged (see Remark \ref{Important} for details). 
%% The cancellation time $\sigma^*$ is thus the first time the seller's price process attains the process $\zeta$.
%}
% 
%
%%
Let $\hat \sigma$ be a stopping time such that $A'_{\hat \sigma}=0$ a.s. and 
 $Y_{\hat \sigma}=\zeta_{\hat \sigma}$ a.s.
By the above proof, the pair $( \hat \sigma,\varphi^*)$ is a super-hedge for the initial amount $u_0$, that is  $( \hat \sigma,\varphi^*)
 \in {\cal S}(u_0).$ 
 For example, under the assumption of  Theorem \ref{superhedging} (that is, the  left-u.s.c. property along stopping times of $-\zeta$),  the stopping time 
 $
\bar \sigma:= \inf \{t \geq 0: \,\,\, A'_t >0\}
$
satisfies these two equalities.
%the conditions $A'_{\bar \sigma}=0$ a.s. and 
% $Y_{\bar \sigma}=\zeta_{\bar \sigma}$ a.s. Indeed, from the definition of $\bar \sigma$,  and the continuity of $A'$ (see Proposition \ref{lusc}), we have 
%$ A'_{\bar \sigma}  =0 $ a.s. 
%Also, for all $t \in [0,T]$, we have 
%$A'_t >  A'_{\bar \sigma} = 0$ a.s. on $\{ t > \bar \sigma\}$.
%Since $A'$ increases only on the set $\{ Y_\cdot = \zeta_\cdot\}$, it follows that 
%$Y_{\bar \sigma}=\zeta_{\bar \sigma}$ a.s.  
% Thus, $( \bar \sigma,\varphi^*)$ is a super-hedge for the initial amount $u_0$.
  Note that  $\bar \sigma \geq \sigma^*$. In general, the equality does not hold.
\end{remark}

 \begin{remark}
 Note that under the assumption of Theorem \ref{superhedging}, there does not necessarily exist a saddle point for the {\em generalized Dynkin game} \eqref{Dynkin}. However, if we suppose additionally that $\xi$ is left-u.s.c. along stopping time, there exists a saddle point. More precisely, in this case, 
by \cite[Theorem 4.7]{DQS2}, the pair $(\tau^*, \sigma^*)$, with $\sigma^*$ defined in \eqref{sigmaetoile} and 
$
\tau^*:= \inf \{t \geq 0: \,\,\, Y_t = \xi_t\}$, is a  saddle point for the {\em generalized Dynkin game} \eqref{Dynkin}, that is, 
  for all $(\tau, \sigma) \in \mathcal{T}^2$ we have
$$ {\cal E}^g_{0, \tau \wedge \sigma^*}[I(\tau, \sigma^*)] \leq Y_0={\cal E}^g_{0, \tau^* \wedge \sigma^*}
[I(\tau^*, \sigma^*)] \leq {\cal E}^g_{0, \tau^* \wedge \sigma}[I(\tau^*, \sigma)],$$
which implies that $\tau^*$ is optimal for the optimal stopping problem
 $\sup_{\tau \in \mathcal{T}} \mathcal{E}^g[I(\tau, \sigma^*)]$. \\
%and $\sigma^*$ is optimal for $\inf_{\sigma \in \mathcal{T}} \mathcal{E}^g[I(\tau^*, \sigma)]$.
%It is thus a  {\em rational} exercise time for the buyer of the American option with terminal time 
%$\sigma^*$ and payoff $I(t, \sigma^*)$ (in the sense of Definition \ref{rational exercise}).
The same properties also hold for the pair $(\bar \tau, \bar \sigma)$ where 
$\bar \tau:=  \inf \{t \geq 0: \, A_t >0\}$. 
\end{remark}

We consider now the general case  when $\zeta$ is only RCLL (as $\xi$). In this case,  the seller's price $u_0$ is still equal to the $g$-value but it does not necessarily allow 
the seller to build a {\em super-hedge} against the option. We introduce the 
 definition of $\varepsilon$-{\em super-hedges}: 
\begin{definition}
For each initial wealth $x$ and for each $\varepsilon>0$, an $\varepsilon$-{\em super-hedge} against the game option is a pair $(\sigma, \varphi)$ of a stopping time $\sigma \in {\cal T}$ and a risky-assets strategy  $ \varphi$ $\in$  ${\mathbb H}^2 \times  {\mathbb H}^2_{\lambda}$ such that
\begin{equation*}%\label{case}
 V_t^{x, \varphi} \geq \xi_t, \,\,0 \leq t \leq \sigma\,\,\, {\rm a.s.} \quad {\rm and }\quad
V^{x, \varphi}_{\sigma}  \geq   \zeta_{\sigma}- 
\varepsilon  \,\,\, {\rm a.s.}
\end{equation*}
In other terms, by investing the initial capital amount $x$ in the market following the risky-assets strategy $\varphi$, the seller is completely hedged before $\sigma$, and at the cancellation time $\sigma$, he is hedged up to an amount of $\varepsilon$.

%$V^{x, \varphi}_{t \wedge \sigma} \geq I(t,\sigma)$, $0 \leq t \leq T$ a.s.\, (or equivalently  
\end{definition}

We prove below  that when $\zeta$ and $\xi$ are only RCLL, the seller's price $u_0$ is equal 
to the $g$-value and that there exits an $\varepsilon$-{\em super-hedge} for the game option.

%We introduce the notion of 
%{\em nearly superhedging price}: the minimum initial amount, if it exists, which allows, for each $\varepsilon>0$, to 
%build an $\varepsilon$-{\em super-hedge}...
%
%the $g$-value $Y_0$ does not necessarily allow the seller to construct 
% a super-hedge strategy $(\sigma, \varphi)$. However, the following property holds.
\begin{theorem}[{\em Seller's price} and $\varepsilon$-{\em super-hedge} of the game option] \label{epsil} 

Suppose that the process $\zeta$ and $\xi$ are only RCLL. 
The seller's price  \eqref{seller's price} of the game option coincides with the {\em $g$-value} of the game option, that is 
$$u_0=\inf_{\sigma \in \mathcal{T} }  \sup_{\tau \in \mathcal{T}} \cal{E}^g_{0,\tau \wedge \sigma}[I(\tau, \sigma)]= 
\sup_{\tau\in \mathcal{T} } \inf_{\sigma \in \mathcal{T}} 
\cal{E}^g_{0,\tau \wedge \sigma}[I(\tau, \sigma)].$$
Let $(Y, Z, K, A,A')$ be the  solution of the 
DRBSDE  
associated with driver $g$ and barriers $\xi, \zeta$. The seller's price is equal to $Y_0$,
% % IMPORTANT ICI
% \footnote{\bf
% Similarly, for each $t \in [0,T]$, the seller's price at time $t$ is equal to $Y_t$ (see Remark \ref{Important}).
% }
%  % JUSQU'ICI
that is
 \begin{equation}\label{new}
 u_0=Y_0.
 \end{equation}
The infimum in \eqref{seller's price} is not nessarily attained. Let  $\varphi^*:=\Phi(Z,K)$ and for each $\varepsilon>0$,
 let 
\begin{equation}\label{sigmaep}
\sigma_{\varepsilon}:= \inf \{t \geq 0: \,\,\, Y_t \geq \zeta_t-\varepsilon \}.
\end{equation} 
% Let  $\varphi^*:=\Phi(Z,K)$ (where $\Phi$ is defined in Definition \ref{stbis}). We have
% \begin{equation*}%\label{case}
% V_t^{Y_0, \varphi^*} \geq \xi_t, \,\,0 \leq t \leq \sigma_{\varepsilon}\,\,\, {\rm a.s.} \quad {\rm and }\quad
%V^{Y_0, \varphi^*}_{\sigma_{\varepsilon}}  \geq   \zeta_{\sigma_{\varepsilon}}- 
%\varepsilon  \,\,\, {\rm a.s.}
%\end{equation*}
The pair  
$(\sigma_{\varepsilon}, \varphi^*)$ 
is an 
 $\varepsilon$-super-hedge  for the initial capital $u_0$. 
 \end{theorem}
%\begin{remark}
%This result generalizes the one proven in \cite{H}
% in the particular case of a perfect market model and continuous processes $\xi$ and $\zeta$.
%\end{remark}
 
%is $K \varepsilon$ for \eqref{prixS}, that is 
%$Y_0 \leq  \mathcal{E}_{0, \tau }^g [\xi_{\nu_{\varepsilon}}] + K\varepsilon$, where 
%$K$ is a constant which only depends on $T$ and the Lipschitz constant of $g$.
\begin{proof} By Proposition \ref{fairpricegame}, the {\em $g$-value} is equal $Y_0$. Let $\varepsilon>0$. We have  $Y_{.} \leq \zeta_{.} -\varepsilon$ on $[0 , \sigma_{\varepsilon}[$. 
Since $A'$ satisfies the Skorohod condition (iii), it follows that almost surely, $A'$ is 
constant on $[0, \sigma_{\varepsilon}[$.
 Also, 
$ Y_{ (\sigma^{\varepsilon})^-  } \leq \zeta_{ (\tau^{\varepsilon})^-  }- \varepsilon \,$ a.s.\,, 
which implies that $\Delta A' _ {\sigma_{\varepsilon} } =0$ a.s.\,
Hence, $A' _ {\sigma^{\varepsilon} } =0$ a.s.\, It follows that for almost every $\omega$, 
the deterministic function $Y_.(\omega)$ is the solution of the {\em forward} deterministic  differential equation  \eqref{forwardb}  on $[0, \sigma_{\varepsilon}(\omega)]$. 
Now, for almost every $\omega$, the wealth  $V_.^{Y_0, \varphi ^*}(\omega)$ is 
the solution of the deterministic  differential equation  \eqref{ri}.
By applying the classical comparison result on differential equations (Lemma \ref{classique}), we derive that
$V_t^{Y_0, \varphi ^*} \geq Y_t\geq \xi_t$, $0 \leq t \leq \sigma_{\varepsilon}$ a.s.\, Moreover, we have
  $V_{\sigma_{\varepsilon}}^{Y_0, \varphi ^*} \geq Y_{\sigma_{\varepsilon}}  \geq   \zeta_{\sigma_{\varepsilon}}- 
\varepsilon$, where the last inequality follows from definition of the stopping time $\sigma_{\varepsilon}$  and the 
  right-continuity of $Y$ and $\zeta$. Hence, $(\sigma_{\varepsilon}, \varphi^*)$ 
is an 
 $\varepsilon$-super-hedge  for the initial capital amount $Y_0$.
 
 It remains to show that $Y_0=u_0$. 
 The proof of the inequality $u_0\geq Y_0$, which uses Proposition \ref{fairpricegame}, has been done in the second part of the proof of Theorem \ref{superhedging} and does not require the continuity of $A'$. Let us show the converse inequality. Let $\varepsilon >0$. 
 Let $(Y',Z',K')$ be the solution of the BSDE associated with terminal time $\sigma_{\varepsilon}$ and terminal condition $\zeta_{\sigma_{\varepsilon}}\vee V^{Y_0, \varphi^*}_{\sigma_{\varepsilon}}$. Now $(V^{Y_0, \varphi^*}, Z,K)$ is the solution  of the BSDE associated with terminal time $\sigma_{\varepsilon}$ and terminal condition $V^{Y_0, \varphi^*}_{\sigma_{\varepsilon}}$. 
 By an a priori estimate on BSDEs with default jump (see \cite[Proposition 2.4]{DQS4}), 
 since $V^{Y_0, \varphi^*}_{\sigma_{\varepsilon}}  \geq   \zeta_{\sigma_{\varepsilon}}\vee V^{Y_0, \varphi^*}_{\sigma_{\varepsilon}}- 
\varepsilon$  a.s.\,, we derive that $V^{Y_0, \varphi^*}_{0}=Y_0  \geq   Y'_0 - 
K \varepsilon$ a.s.\,, where $K$ is a constant which only depends on $T$ and the $\lambda$-constant $C$.
%and $V^{Y_0, \varphi^*}_{t} \geq   Y'_t - 
%K \varepsilon$ a.s. ?? 
By the comparison theorem for BSDEs, $Y'_t \geq V^{Y_0, \varphi^*}_{t}\geq \xi_t$. 
We derive that the amount $Y'_0$ ($\leq Y_0 + K \varepsilon$) allows the seller to be super-hedged, and 
 the associated super-hedge is given by $\sigma_{\varepsilon}$ and $\varphi ' := \Phi (Z',K')$. By definition of $u_0$, we derive that
 $u_0 \leq Y'_0 \leq Y_0 + K \varepsilon$, for each $\varepsilon>0$. Hence, $u_0 \leq Y_0$. Since $u_0 \geq Y_0$, we get $u_0 = Y_0$. 
 \end{proof}
\section{Pricing and hedging of game options with model uncertainty}\label{mixed}

%Note that the  superhedging problem of a game option in a continuous time framework has not been studied in the literature, even in the case of a perfect market. 

We study now game options  with uncertainty on the model, which includes in particular the case of  uncertainty on the default 
probability (see  Example \ref{example} below).

\subsection{Market model with ambiguity}\label{mamo}

%Our aim is to obtain  the seller's robust price of the game option under uncertainty as the value function 
%of a mixed control game problem, as well as its characterization as the
%solution of an associated DRBSDE.
 %To this purpose, we
 In this section, we need to  use  a measurable selection theorem, 
which requires to work on an appropriate probability space.
  We consider a Cox process model, which is a typical example of default model. 
 We  work on  the canonical space constructed as follows:
 let  $\Omega_W$ be the Wiener space defined by  $\Omega_W:= {\cal C}  ({\mathbb R}^+)$, that is the set of continuous functions $\omega$ from ${\mathbb R}^+$ into $\mathbb{R}$ such that $\omega(0) = 0$. Recall that $\Omega_W$ is a Polish space for the
 norm $\|\cdot \|_\infty$.    The space $\Omega_W$ is equipped with the 
 $\sigma$-algebra ${\cal F}_W$ generated by the coordinate process $(W_t)_{t\geq 0}$ (which is equal to its Borelian $\sigma$-algebra). 
% (canonical) Wiener $(\Omega_W, {\cal F}_W, P)$?\\
 Let $P_W$ be the probability under which $(W_t)_{t\geq 0}$ is a standard Brownian motion. 
 Let $\Omega_\Theta := \mathbb{R}$, equipped with its Borelian $\sigma$-algebra 
 ${\cal F}_\Theta= {\cal B}(\mathbb{R})$, and the
  probability $P_\Theta$ such that the identity map $\Theta$ admits an exponential law with parameter $1$.  
  We consider the product space $\Omega:=\Omega_W \times \Omega_\Theta$, which is a Polish space. It is equipped with the 
  $\sigma$-algebra ${\cal F}_W \otimes {\cal F}_\Theta$, and the probability $P:=P_W \otimes P_\Theta$. 
  Let ${\cal G}$ be the  $\sigma$-algebra ${\cal F}_W \otimes {\cal F}_\Theta$ completed with respect to $P$.  Let ${\mathbb F}=(\mathcal{F}_t,t\geq 0)$  be the filtration ${\cal F}_W$ {\em completed}  with respect to ${\cal G}$ and $P$ (in the sense of \cite[p.3]{J}  or \cite[IV]{DM1}). Let $(\bar \lambda_t)_{t\geq 0}$ be a bounded positive 
  $\mathbb F$-predictable process. We introduce the following random variable, which represents the {\em default time}:    
  $$\vartheta:= \inf \{t \geq 0, \,\, \int_0^t \bar\lambda_s ds \geq \Theta \}.$$
  We have $P( \vartheta >t \,|\, {\cal F}_{\infty}) = 
  P( \vartheta >t \,|\, {\cal F}_{t})= \exp (-\int_0^t \bar\lambda_s ds)$, 
 which corresponds to the so-called condition ({\bf H}) (see e.g. \cite{JYC}).
 We now define the  {\em default process}:
$$N_t\,:=\,{\bf 1}_{\{\vartheta \leq t\}}\,, \quad t\geq 0.$$
%Revuz-Yor p34 espace polonais. 
% Let $B= (B^1, B^2)$ be the canonical process defined for each $t \in [0,T]$ and each $\omega= (\omega^1, \omega^2)$ by $B^i_t(\omega)= B^i_t(\omega^i):=\omega^i_t$, for $i=1,2$.   
%Let us denote the first coordinate process $B^1$ by $W$. 
%Let $P^W$ be the probability measure on  $(\Omega_W,\mathcal{B}(\Omega_W))$ such that $W$ is a Brownian motion. Here $\mathcal{B}(\Omega_W)$ denotes the Borelian $\sigma$-algebra on $\Omega_W$.
%
%
%Let $(\Omega, \mathcal{F},P)$ be a probability space.
%% We assume that all processes are defined on a finite time horizon $[0,T]$, 
%% with $T < \infty$,  and 
% We suppose the space to be equipped with two stochastic processes:
%  a unidimensional standard Brownian motion $W$ and a jump process $N$ defined by 
%  $N_t={\bf 1}_{\vartheta \leq t}$ for any $t\geq 0$, where $\vartheta$ is a random variable which represents a default time. We assume that this default can appear at any time that is $P(\vartheta \geq t)>0$ for any $t\geq 0$.
   We denote by ${\mathbb G}=(\mathcal{G}_t,t\geq 0)$  the  filtration generated by  $W$ and $N$ {\em augmented} with respect to ${\cal G}$ and $P$ (in the sense of \cite[IV-48]{DM1}).
By classical results, since Condition ({\bf H}) holds, we derive that  $W$ is a ${\mathbb G}$-Brownian motion. 
 Moreover, the process $M$ defined by
% the compensated $\mathbb G$-martingale of the process $N$, denoted by $M$, %  and by 
% $\Lambda$ its ${\cal G}$-compensator which is assumed to be absolutely continuous w.r.t. Lebesgue's measure, 
% so that 
% there exists an $\mathbb F$-predictable process $\lambda$, 
% called intensity process, such that
%satisfies the equality
 \begin{equation*}
%\label{M}
M_t  := N_t-\int_0^{t \wedge \vartheta }\bar \lambda_sds, \quad  t\geq 0, \quad {\rm a.s.}
\end{equation*}
is a $\mathbb G$-martingale. 
%In other terms, $M$ is the 
%compensated $\mathbb G$-martingale of the process $N$.
For each $t\geq 0$, let  $ \lambda_t:= \bar \lambda_t \, {\bf 1}_{\{t \leq \vartheta \}}$.
%0$ for $s> \vartheta$, and $\bar \lambda_s:=\lambda_s$ for $s\leq \vartheta$. 
The process $\lambda$, usually called the ${\mathbb G}$-{\em intensity} of $\vartheta$, thus  vanishes after $ \vartheta$. 
Let $T$ be a given terminal time. The sets ${\cal P}$,  ${S}^2$, ${\mathbb H}^2$, ${\mathbb H}^2_{\lambda}$ and ${\cal A}^2$ are defined as before. 
%JUSQU ICI

%An example where there is some ambiguity on the {\em probability of default} is given below. 
Let  $U$ be a nonempty closed subset of $\mathbb{R}$.
Let  
$g:[0,T] \times  \Omega \times \R^3 \times U \rightarrow \R $ ; 
$(t, \omega, y,z, k,\alpha) \mapsto  g(t, \omega, y,z, k, \alpha),$
 be a given   $ {\cal P} \otimes {\cal B}(\R^3) \otimes {\cal B}(U)$-measurable function.
Suppose $g(\cdot, \alpha)$ is  {\em uniformly $\lambda$- admissible}  with respect to $(y,z, k)$, that satisfies the inequality \eqref{lip} with a constant $C$ which does not depend on $ \alpha$. We also assume that $g(\cdot, \alpha)$ is continuous with respect to $\alpha$,
 and such that 
$ \sup_{\alpha \in U}  | g(  t, .,0,0,0, \alpha)|  \in \H_{2}$. Suppose also that
 \begin{equation}\label{rrrbis}
g(t,y,z,k_1, \alpha)- g(t,y,z,k_2, \alpha) \geq 
 \theta_t^{y,z,k_1,k_2} (k_1 - k_2 )  \lambda_t,
\end{equation}
where  
$\theta_t^{y,z,k_1,k_2}$ satisfies the conditions of Assumption \ref{Royer}, in particular the inequality $\theta_t^{y,z,k_1,k_2}>-1$. 
%We also suppose that $g(t,0,0,0, \alpha)\geq 0$ for all $\alpha \in U$.
%\footnote{
%%
%{\bf For example, $G$ can be given by  $G(t,y,z,k,  \alpha):= - r_t y - (z + \sigma^2_t k\,{\bf 1}_{\{t \leq \vartheta \} }) \theta^1_t -  \theta^2_t \lambda_t k +  \nu(t,\alpha) \lambda_t k$
%(where $\nu$ satisfies $\nu (t,\alpha)> C_1   >-1$). In this case, the coefficient  $\nu(t,\alpha)$ can represent some uncertainty on the default intensity (see Example \ref{example} for details).
%}
%}

 %with
%\begin{equation*}
%\tau:  \Omega \times [0,T]\times \RB \times ({\cal L}_\nu^2)^2 \times  A \mapsto  {\cal L}_\nu^2
%\end{equation*}
% ${\cal P } ·\otimes {\cal B}(\R) \otimes  {\cal B}({\cal L}_\nu^2)^2  \otimes  {\cal B}(A)$-measurable and
%satisfying
%$
%|\theta_t^{z,l_1,l_2, \alpha}(.) | \leq \bar \psi (.),
%$
%where $\bar \psi$ belongs to $L^p_{\nu}$, for all $p \geq 2$, and 
%$\theta_t^{z,l_1,l_2, u} \geq  -1$. 
%Without loss of generality, we can suppose that  $C_2$ is a Lipschitz constant for $F$. 
 % Suppose that  $F, \beta^1, \beta^2$ are continuous with respect to $\alpha$. 
 
Let ${\cal U}$ be the set of $U$-valued predictable processes. 
For each $\alpha \in {\cal U}$, to simplify notation, we introduce the map $g^{\alpha}$ defined by
\begin{equation}\label{baralpha}
 g^{\alpha} (t, \omega, y,z, k):= g(t, \omega, y,z, k, {\alpha}_t(\omega)).
 \end{equation}
Note that these maps $g^{\alpha}$, $\alpha \in {\cal U}$, are all $\lambda$-{\em admissible} drivers with the same $\lambda$-constant $C$.
The control $\alpha$ represents the ambiguity parameter of the model. To each ambiguity parameter $\alpha$, corresponds a market model ${\cal M}_{\alpha}$ where the  
wealth process $V^{\alpha ,x, \varphi}$ associated with an initial wealth $x$ and  a risky assets stategy $\varphi$ $\in$ ${\mathbb H}^2 
\times {\mathbb H}^2_{\lambda}$ satisfies 
 \begin{eqnarray}\label{richessealpha}
-dV^{\alpha ,x, \varphi}_t= g(t,V^{\alpha ,x, \varphi}_t,  \varphi_t \sigma_t ,-\varphi_t^{2} ,{\alpha}_t) dt -  \varphi_t \sigma_t dW_t +\varphi_t^{2}dM_t\,;\,\,\,
V^{\alpha ,x, \varphi}_0=x.
\end{eqnarray}
 In the market model ${\cal M}_{\alpha}$, the nonlinear pricing system is given by 
${\cal E}^{g^{\alpha}}:= \{{\cal E}_{t, S}^{g^{\alpha}}, \,\, S \in [0,T], t \in [0,S]\}$, also called $g^{\alpha}$-evaluation. 

\subsection{Robust superhedging of game options}
In our framework with ambiguity, the {\em seller's robust price} of the game option denoted by ${\bf u_0}$ is defined as the infimum of the initial wealths which enable the seller to be superhedged for %each
  any  ambiguity parameter  $\alpha \in {\cal U}$.
  %
%In our model with ambiguity, the {\em seller's robust price} of the game option denoted by ${\bf u_0$} is defined as the infimum of the initial wealths which enable the seller to choose 
%a cancellation time $\sigma$  and to 
%invest in a portfolio which, for %each
%  any  ambiguity parameter  $\alpha \in {\cal U}$, will cover  his liability to pay the payoff to the buyer up to $\sigma$ no matter what exercise time the buyer chooses. 
  
 \begin{definition}\label{definitionrobust}
  For an initial wealth $x \in {\mathbb R}$, a {\em robust super-hedge} against the game option is a pair $(\sigma, \varphi)$ of a stopping time $\sigma \in {\cal T}$ and a portfolio strategy $ \varphi$ $\in$  ${\mathbb H}^2 \times  {\mathbb H}^2_{\lambda}$ such that  \footnote{Condition \eqref{condAbis} is equivalent to $V^{\alpha ,x, \varphi}_{t \wedge \sigma} \geq I(t,\sigma)$, $0 \leq t \leq T$ a.s.\, for all $\alpha \in {\cal U}$.}
    \begin{equation}\label{condAbis}
V^{\alpha, x, \varphi}_{t } \geq \xi_t,  \; 0\leq t \leq \sigma \; \text{ a.s. and } V^{\alpha, x, \varphi}_{\sigma } \geq \zeta_{\sigma}
 \text{ a.s.}\,, \quad \forall \alpha \in {\cal U}.
\end{equation}
We denote by     ${\cal S}^r (x)$ the {\em set of all robust super-hedges associated with initial wealth}
 $x$. 
 
The {\em seller's robust price} is defined as
\footnote{
Remark \ref{positive2} also holds for the {\em seller's robust price}, that is, ${\bf u_0} \leq \zeta_0$. Moreover, when $g(t,0,0,0)=0$ and $\zeta \geq 0$, then
${\bf u_0}= \inf \{x \geq 0,\,\, \exists  (\sigma, \varphi) \in {\cal S}^r  (x) \}$.}
\begin{equation}\label{robustprice}
 {\bf u_0}:= \inf \{x \in \R,\,\, \exists  (\sigma, \varphi) \in {\cal S}^r  (x) \}.
 \end{equation}
When the infimum is reached, ${\bf u_0}$ is called the {\em robust superhedging price}.
\end{definition}

Let $\alpha \in \mathcal{U}$.  By Theorem \ref{epsil}, the seller's price of the game option in the market 
${\cal M}_\alpha$ is characterized as its  {\em $g^\alpha$-value}. Moreover,  it is equal to $Y^\alpha_0$, where 
$(Y^\alpha, Z^\alpha, K^\alpha, A^\alpha,A^{' \alpha})$ is the unique solution in  ${S}^2 \times {\mathbb H}^2 \times {\mathbb H}^2_{\lambda}\times {\cal A}^2 \times {\cal A}^2$ 
of the DRBSDE  
associated with driver $g^\alpha$ and barriers $\xi$ and $ \zeta$. 
We now introduce an associated dual problem.
\begin{definition}
The {\em dual problem} associated to the seller's super-hedging problem is
%{\bf A possible price $Y(0)$ of the game option is given by the supremum over $\alpha \in {\cal U}$ of the (superhedging) prices $Y^\alpha _0$ corresponding to the market models ${\cal M}_{\alpha}$, that is }
 \begin{equation} \label{q}
{\bf v_0}: = \sup_{\alpha \in \mathcal{U}} Y^\alpha _0.
%= \sup_{\alpha \in \mathcal{U}} \inf_{\sigma \in \mathcal{T} }  \sup_{\tau \in \mathcal{T}} \cal{E}^{g^\alpha}_{0,\tau \wedge \sigma}[I(\tau, \sigma)]
%= \sup_{\alpha \in \mathcal{U}}   \sup_{\tau \in \mathcal{T}} \inf_{\sigma \in \mathcal{T} } \cal{E}^{g^\alpha}_{0,\tau \wedge \sigma}[I(\tau, \sigma)],
\end{equation}
\end{definition}
By Theorem \ref{epsil},  the seller's price $Y^\alpha_0$ of the game option in the market 
${\cal M}_\alpha$ is equal to the common value function of the  {\em generalized Dynkin game} associated with driver $g^\alpha$, that is, 
$$Y^\alpha _0=
 \inf_{\sigma \in \mathcal{T} }  \sup_{\tau \in \mathcal{T}} \cal{E}^{g^\alpha}_{0,\tau \wedge \sigma}[I(\tau, \sigma)]
=  \sup_{\tau \in \mathcal{T}} \inf_{\sigma \in \mathcal{T} } \cal{E}^{g^\alpha}_{0,\tau \wedge \sigma}[I(\tau, \sigma)].$$
Hence, the value function ${\bf v_0}$ of the dual problem is equal to the value function of a {\em mixed generalized Dynkin game}, that is
\begin{equation} \label{q2}
{\bf v_0} = \sup_{\alpha \in \mathcal{U}} Y^\alpha _0
= \sup_{\alpha \in \mathcal{U}} \inf_{\sigma \in \mathcal{T} }  \sup_{\tau \in \mathcal{T}} \cal{E}^{g^\alpha}_{0,\tau \wedge \sigma}[I(\tau, \sigma)]
= \sup_{\alpha \in \mathcal{U}}   \sup_{\tau \in \mathcal{T}} \inf_{\sigma \in \mathcal{T} } \cal{E}^{g^\alpha}_{0,\tau \wedge \sigma}[I(\tau, \sigma)].
\end{equation}
\begin{remark} 
We shall see below (see Proposition \ref{interversion}) that ${\bf v_0}$ is also equal to:
 $$\inf_{\sigma \in \mathcal{T} }\sup_{\alpha \in \mathcal{U}}  \sup_{\tau \in \mathcal{T}}
 \cal{E}^{g^\alpha}_{0,\tau \wedge \sigma}[I(\tau, \sigma)]
 .$$
\end{remark}
In order to show that ${\bf u_0}={\bf v_0}$, we will first prove that 
 ${\bf v_0}$ %the {\em $g$-value} $Y(0)$ of the game option 
can be 
characterized as the solution of a doubly reflected BSDE. 

Now, by definition, we have
${\bf v_0} = \sup_{\alpha } Y^\alpha _0$, where $Y^\alpha$ is the solution of the doubly reflected BSDE associated with barriers $\xi$ and $\zeta$, and with driver $g(\cdot, \alpha_t)$. We will show that ${\bf v_0} $ coincides with the solution of the doubly reflected BSDE associated with the same barriers $\xi$ and $\zeta$, and with the driver $\sup_{\alpha}g(\cdot, \alpha)$.

%Note  that ${\bf v_0}$ is defined as the supremum over $\alpha \in \mathcal{U}$ of $Y^\alpha _0$, where $Y^\alpha$ is the solution of the doubly reflected BSDE associated with barriers $\xi$ and $\zeta$, and with driver $g(\cdot, \alpha_t)$. 
%We will show that $\sup_{\alpha \in \mathcal{U}} Y^\alpha _0$ (which is equal to ${\bf v_0}$) coincides with the solution of the doubly reflected BSDE associated with the same barriers $\xi$ and $\zeta$, and with the driver ${\bf G}$ defined as the supremum over $\alpha$ of  $g(\cdot, \alpha)$.

More precisely, let ${\bf G}$ be the map defined for each $(t, \omega, z, k)$ by
\begin{equation}\label{d}
{\bf G}(t, \omega, y,z, k):=\sup_{\alpha \in U} g(t, \omega, y,z, k, \alpha) .
\end{equation}

\begin{lemma}\label{G}
The map  ${\bf G}$ is a $\lambda$-{\em admissible} driver and satisfies 
Assumption \ref{Royer}.
%(see Lemma $\ref{G}$).
\end{lemma}
\begin{proof}Since $U$ is a closed subset of a Polish space, there exists a numerable subset $D$ of $U$, dense in $U$. Since $g$ is continuous with respect to $u$, 
the supremum in \eqref{d} can be taken in $D$. It follows that ${\bf G}$ is $ {\cal P} \otimes {\cal B}(\R^3)-$ measurable. 
Let us show that ${\bf G}$ satisfies Assumption $\eqref{Royer}.$  By  definition of ${\bf G}(t,y,z,k_1)$  % (see $\eqref{d}$) 
and by  Assumption $\eqref{rrrbis}$, we have for all $\alpha \in  \mathcal{U}$:
$${\bf G}(t,y,z,k_1)-g(t,y,z,k_2,\alpha)\geq g(t,y,z,k_1, \alpha)-g(t,y,z,k_2,\alpha)\geq  \theta_t^{y,z,k_1,k_2} (k_1 - k_2 )\lambda_t.$$
Taking the infimum on $\alpha \in \mathcal{U}$ in this inequality, and using the definition of ${\bf G}(t,y,z,k_2)$, we derive that  
${\bf G}(t,y,z,k_1)- {\bf G}(t,y,z,k_2) \geq 
 \theta_t^{y,z,k_1,k_2} (k_1 - k_2 )\lambda_t$, which gives the desired result.
 The proof of condition  \eqref{lip}   relies on similar arguments and is left to the reader. Hence, ${\bf G}$ is a $\lambda$-admissible driver.
 \end{proof}

We now prove that the dual function
 ${\bf v_0}$ %the {\em $g$-value} $Y(0)$ of the game option 
is 
characterized as the solution of the doubly reflected BSDE associated with  
driver ${\bf G}$  and barriers $\xi$ and $\zeta$.

% \begin{proposition}\label{2principle} (Characterization of the {\em $g$-value})
%  Let $g$ be a Lipschitz driver.
% Let  \\
% $(Y, Z ,K, A, A')$ be the solution of the DRBSDE associated with  
%driver $g$  and barriers $\xi$ and $\zeta$.\\
%Suppose that for each $u$ $\in$ ${\cal U}$, 
%$g \geq g^{\alpha }$ and
%that, for each $\eta >0$ , there exists $u^{\eta}$ $\in$ ${\cal U}$ such that
%\begin{equation}\label{ep1}
% g(t,Y_t,Z_t,K_t) \leq g^{{u^{\eta}}}(t,Y_t,Z_t,K_t)+ \eta,  \, \;\   0 \leq t \leq  T, \,
% dP\otimes dt-{\rm a.s.}
%\end{equation}
%We then have
%\begin{equation}\label{q}
%Y_0= \sup_{\alpha \in {\cal U}} Y^{\alpha }_0 (= Y(0)).
%\end{equation}
%MOROEVER, IF $U$ is COMPACT ... there exists ${\bar \alpha}$ $\in$ ${\cal U}$ optimal for \eqref{q}, that is such that 
%$Y_0 =Y^{{\bar \alpha}}_0 $. au lieu de\\
%Moreover, if there exists ${\bar \alpha}$ $\in$ ${\cal U}$ such that
%\begin{equation}\label{existalp}
%f(t,Y_t, Z_t,K_t) = %{\rm ess} \inf_{\alpha}f^{\alpha}(t,Y_t, Z_t,K_t) = 
%f^{{\bar \alpha}}(t,Y_t, Z_t,K_t),  \; \; 0 \leq t \leq T, \;\;   dt \otimes dP-{\rm a.s.}
%\end{equation}
%Then, ${\bar \alpha}$ is optimal for \eqref{q}, that is
%$Y_0 =Y^{{\bar \alpha}}_0 $, which means that the {\em fair price} of the game option is equal to the {\em fair price} in the market model 
%${\cal M}_{{\bar \alpha}}$.
%\end{proposition}
 \begin{theorem} (Characterization of the dual value function ${\bf v_0}$) \label{cha}
 Let  ${\bf v_0}$ be 
 %$Y(0)$ be the {\em $g$-value} of the game option 
 defined by \eqref{q}. \\
 We have ${\bf v_0}= Y_0$, where
  $(Y, Z ,K, A, A')$ be the solution of the DRBSDE associated with  
driver ${\bf G}$  and barriers $\xi$ and $\zeta$.
 %   Let  \\
% $(Y, Z ,K, A, A')$ be the solution of the DRBSDE associated with  
%driver $g$  and barriers $\xi$ and $\zeta$.\\
%Suppose that for each $u$ $\in$ ${\cal U}$, 
%$g \geq g^{\alpha }$ and
%that, for each $\eta >0$ , there exists $u^{\eta}$ $\in$ ${\cal U}$ such that
%\begin{equation}\label{ep1}
% g(t,Y_t,Z_t,K_t) \leq g^{{u^{\eta}}}(t,Y_t,Z_t,K_t)+ \eta,  \, \;\   0 \leq t \leq  T, \,
% dP\otimes dt-{\rm a.s.}
%\end{equation}
%We then have
%In the case where 
If $U$ is %supposed to be
 compact, there exists ${\bar \alpha} \in {\cal U}$ such that 
${\bf v_0} =Y^{{\bar \alpha}}_0 $, which means that the 
dual value function ${\bf v_0}$ is equal to the {\em $g^{{\bar \alpha}}$-value} of the game option in the market model 
${\cal M}_{{\bar \alpha}}$.   
\end{theorem}

\begin{proof} 
By definition of ${\bf G}$ (see \eqref{d}),  for each $(t, \omega, y,z, k)$ $\in$ $[0,T] \times  \Omega \times \R^3  \times U $, we have
   \begin{equation*}\label{gG}
   {\bf G}(t, \omega, y,z, k) \geq g(t, \omega, y,z, k, {\alpha}_t(\omega)).
%   = g^{\alpha}(t,\omega, y,z, k).
    \end{equation*}
By the comparison theorem for DRBSDEs (see Theorem 5.1 in \cite{DQS2}), we thus have $Y \geq Y^{\alpha}$ a.s.\, 
 for each $\alpha \in {\cal U}$. It follows that  $Y_0 \geq \sup_{\alpha} Y^{\alpha }_0$.\\
 Let   $\varepsilon>0$. By definition of ${\bf G}$ as a supremum, for each $(t, \omega, y, z, l)$ $\in$ $\Omega \times [0,T] \times \R^2 \times \mathbb{R}$, there exists $\alpha ^{\varepsilon}$ $\in$ $U$ such that 
${\bf G} (t, \omega, y,z, k) -\varepsilon \leq  g (t, \omega, y,z, k, \alpha ^{\varepsilon}) . $
Now, the set
$$\{(t,\omega,\alpha) \in [0,T] \times \Omega \times U:\,\, {\bf G}(t,\omega,Y_{t^-}(\omega), 
Z_t(\omega), K_t(\omega)) -\varepsilon \leq g (t, \omega,Y_{t^-}(\omega), Z_t(\omega), K_t(\omega), \alpha)\} $$ 
belongs to $\mathcal{P} \otimes \mathcal{B}(U)$. 
Hence, since the canonical space $\Omega$ is a Polish space, by applying 
 a measurable selection theorem (see e.g.
 \cite[Section 81, Appendix of Ch. III]{DM1}) and  \cite[Lemma 1.2]{C} (or \cite[Lemma 26]{DQS1}), there exists an $U$-valued predictable process $( \alpha ^{\varepsilon}_t)$ such that 
\begin{equation*}
 {\bf G} (t,  Y_t,Z_t, K_t)-\varepsilon \leq   g (t, \omega, Y_t,Z_t, K_t , \alpha ^{\varepsilon}_t),  \; \; 0 \leq t \leq T, \;\;   dt \otimes dP-{\rm a.s.}
\end{equation*}
By using the estimate \eqref{A26} on DRBSDEs with default jump, with $\eta = \frac{1}{C^2}$ and $\beta = 3C^2 + 2C$, we derive 
that there exists a constant $K\geq 0$, which depends only on $C$ and $T$, such that, for each $\varepsilon>0$,
$$Y_0 - K \,\varepsilon \leq Y^{{\alpha^{\varepsilon }}}_0 .
$$ 
Since $Y_0\geq \sup_{\alpha} Y^{\alpha }_0 $, we thus get  $Y_0= \sup_{\alpha} Y^{\alpha }_0= {\bf v_0}$.\\
 Let us show the second assertion. If $U$ is compact, for each $(t, \omega, y, z, l)$ $\in$ $ [0,T] \times \Omega \times \R^2 \times L^2_\lambda$, 
  there exists ${\bar \alpha}$ $\in$ $U$ such that the supremum in \eqref{d} is attained at ${\bar \alpha}$. By the measurable selection theorem of \cite{DM1} and  \cite[Lemma 1.2]{C},
 there exists an $U$-valued predictable process $( {\bar \alpha}_t)$ such that 
\begin{equation*}
{\bf G}(t,Y_t, Z_t,K_t) = %{\rm ess} \inf_{\alpha}f^{\alpha}(t,Y_t, Z_t,K_t) = 
g(t,Y_t, Z_t,K_t, {{\bar \alpha}}_t),  \; \; 0 \leq t \leq T, \;\;   dt \otimes dP-{\rm a.s.}
\end{equation*}
 It follows that $Y$ and $Y^{{\bar \alpha}}$ are both solutions of the DRBSDE associated 
 with driver $g^{{\bar \alpha}}$. Hence, by the uniqueness of the solution of a DRBSDE, $Y$ $=Y^{{\bar \alpha}}$.
 \end{proof}
Using this result, we now provide the following theorem:
\begin{theorem}\label{srp}(Seller's robust price and  super-hedge) 
Suppose that $\zeta$ is left-lower semicontinuous along stopping times (and $\xi$ is only RCLL). 
 The {\em seller's robust price}  of the game option defined by \eqref{robustprice} is equal the dual
  value function ${\bf v_0}$ defined by \eqref{q}, that is 
$${\bf u_0} = {\bf v_0}.$$
Let  $(Y, Z ,K, A, A')$ be the solution of the DRBSDE associated with  
driver ${\bf G}$  defined by \eqref{d} and barriers $\xi$ and $\zeta$. The seller's robust price is equal to $Y_0$, that is
 $${\bf u_0}=Y_0.% IMPORTANT A CARDER
 $$
Moreover, the infimum in \eqref{q} is attained. The robust seller's price is thus the {\em robust super-hedging price}
 of the game option.
Let  $\sigma^*:= \inf \{ t \geq 0,\,\, Y_t = \zeta_t\}$ and $\varphi^*:=\Phi(Z,K)$. The pair $(\sigma^*, \varphi^*)$  is a robust super-hedge for the initial capital ${\bf u_0}$. 

% IMPORTANT A CARDER
 If $U$ is compact,  there exists ${\bar \alpha}$ $\in$ ${\cal U}$ such that the robust superhedging price of the game option is equal to the superhedging price in the market model 
${\cal M}_{{\bar \alpha}}$, that is ${\bf u_0} = Y_0^{\bar \alpha}$. The ambiguity parameter  $ {{\bar \alpha}}$ corresponds to a {\em worst case scenario}
among  all the possible ambiguity parameters $\alpha \in \cal U$.
%\begin{align}
%{\bf u_0}=Y_0=\sup_{\tau \in \mathcal{T}} \inf_{\sigma \in \mathcal{T}} \mathcal{E}_{0,T}^f[I(\tau, \sigma)]=\inf_{\sigma \in \mathcal{T}} \sup_{\tau \in \mathcal{T}} \mathcal{E}_{0,T}^f[I(\tau, \sigma)].
%\end{align}
\end{theorem}

\begin{proof} By Theorem \ref{cha}, ${\bf v_0} = Y_0$. 
 Let $\mathcal{H}^r$  be the set of initial capitals which allow the seller to be {\em super-hedged}, that is
$\mathcal{H}^r= \{ x \in \mathbb{R}: \exists ( \sigma,\varphi) \in {\cal S}^r (x) \}. $
Note that ${\bf u_0}= \inf \mathcal{H}^r.$

Let us show that $Y_0 \geq {\bf u_0}$. It is sufficient to show that there exists 
$( \sigma^*,\varphi^*) \in {\cal S}^r (Y_0)$. By Proposition \ref{lusc}, since $- \zeta$ is left-u.s.c. along stopping times, the process  
$A'$ is continuous.
%Let  $\sigma^*:= \inf \{ t \geq 0,\,\, Y_t = \zeta_t\}.$
By definition of $\sigma^*$, the process $A'$ is constant on $[0, \sigma^*[$ a.s.\, and even on $[0, \sigma^*]$ by continuity.
%Let  $\sigma^*:= \inf \{ t \geq 0,\,\, Y_t = \zeta_t\}.$
%By definition of $\sigma^*$, we have a.s. that $Y_t< \zeta_t$ for each $t \in$ $[0, \sigma^*[$.  Hence, since $Y$ is solution of the DRBSDE, the continuous process $A'$ is constant on $[0, \sigma^*[$ a.s.\, and even on $[0, \sigma^*]$ by continuity. Hence, 
Hence, $A'_{\sigma^*}=A'_0=0$ a.s.\,
%Since $A'$ is continuous, 
%the non decreasing processes are continuous, there exists a saddle point $(\sigma^*, \tau^*) \in \mathcal{T} \times \mathcal{T}$ for the generalized Dynkin game.
 We thus have
\begin{align*}
Y_t=Y_0-\int_0^t {\bf G}(s,Y_s,Z_s,K_s)ds+\int_0^t Z_s dW_s +\int_0^t K_s dM_s-A_t, \,\, 0 \leq t \leq \sigma^* \quad {\rm a.s.}
\end{align*}
Let $\alpha \in {\cal U}$. In the market model ${\cal M}_{\alpha}$, the wealth process $V_.^{\alpha ,Y_0, \varphi ^*}$ associated with the initial capital $Y_0$ and the financial strategy 
$\varphi^*:=\Phi(Z,K)$ satisfies
\begin{align*}\label{riche}
V_t^{\alpha ,Y_0, \varphi ^*}=Y_0-\int_0^t g(s,V^{\alpha ,Y_0, \varphi ^*}_s,Z_s,K_s, \alpha_s)ds+\int_0^t Z_s dW_s +\int_0^t K_s dM_s.
\end{align*}
By definition of ${\bf G}$ (see \eqref{d}),   we have
%   $g(t, \omega, y,z, k) \geq G(t, \omega, y,z, k, {\alpha}_t(\omega)).
%    $ 
    $-g(t, \omega, y,z, k, {\alpha}_t(\omega)) \geq -{\bf G}(t, \omega, y,z, k).
    $ 
    Hence, since $A$ is a non decreasing process, 
by the comparison property for deterministic differential equations (see Lemma \ref{classique}) applied to the two above forward equations, we derive that
\begin{equation*}%\label{oo}
V_t^{\alpha ,Y_0, \varphi ^*} \geq Y_t  \geq \xi_t, \,\,0 \leq t \leq \sigma^* \quad {\rm a.s.}\,,
\end{equation*}
where the last inequality follows from the inequality $Y \geq \xi$.\\
Moreover, we have
$
  V_{\sigma^*}^{\alpha ,Y_0, \varphi ^*} \geq Y_{\sigma^*} = \zeta_{\sigma^*}$ a.s.\,,  and this holds for any $\alpha \in {\cal U}$. 
% where the last equality follows from
%   the definition of the stopping time $\sigma^*$  and the right-continuity of $Y$ and $\zeta$. 
   Hence 
  $( \sigma^*,\varphi^*) \in {\cal S}^r (Y_0)$, which implies
  $Y_0 \in \mathcal{H}^r$. Thus, $Y_0 \geq {\bf u_0}$. 
  
Let us now show that ${\bf u_0} \geq Y_0$.
%Since $Y_0$ coincides with the value function of the generalized Dynkin game, we have $Y_0=  \sup_{\tau \in \mathcal{T}}  \inf_{\sigma \in \mathcal{T}} \mathcal{E}_{0,T}^f[I(\tau, \sigma)]=\underline{V}(0).$
%
%Since moreover $\underline{V}(0) \leq \overline{V}(0)$, in order to show the converse inequality, i.e. 
%Since $Y_0= \inf_{\sigma \in \mathcal{T}} \sup_{\tau \in \mathcal{T}} \mathcal{E}_{0,T}^f[I(\tau, \sigma)]$,  it is sufficient to show that
%\begin{align}\label{dd}
%{\bf u_0} \geq  \inf_{\sigma \in \mathcal{T}} \sup_{\tau \in \mathcal{T}} \mathcal{E}_{0,T}^f[I(\tau, \sigma)].
%\end{align}
Let $x \in \mathcal{H}^r$. There exists $( \sigma,\varphi) \in {\cal S}^r (x)$, that is 
  a pair $(\sigma, \varphi)$ of a stopping time $\sigma \in {\cal T}$ and a portfolio strategy $ \varphi$ $\in$  ${\mathbb H}^2 \times  {\mathbb H}^2_{\lambda}$ such that for each 
  $\alpha \in {\cal U}$, we have
$V^{\alpha ,x, \varphi}_{t \wedge \sigma} \geq I(t,\sigma)$, $0 \leq t \leq T$ a.s. 
%For each $\tau \in \mathcal{T}$ we thus have
%$$V^{x, \varphi}_{\tau \wedge \sigma} \geq I(\tau, \sigma) \quad {\rm a.s.}$$
%$$V_\sigma^{x, \varphi} \geq \xi_\tau \textbf{1}_{\{\tau < \sigma\}} + \zeta_\sigma \textbf{1}_{\{\tau=\sigma\}}$$
%By taking the $\mathcal{E}^g$-evaluation in the above inequality and then the supremum on 
%$\tau \in \mathcal{T}$, using the $\mathcal{E}^g$-martingale property of the wealth process $V^{x, \varphi}$,
%we obtain
%\begin{align}
%x =\mathcal{E}^g _{0,\tau \wedge \sigma}[V^{x, \varphi}_{\tau \wedge \sigma}] \geq \sup_{\tau \in \mathcal{T}}
%\mathcal{E}_{0, \tau \wedge \sigma}^g [ I(\tau, \sigma)].
%\end{align}
%Hence, 
By the same arguments as in the proof of Theorem \ref{superhedging}, we derive that for each $\alpha \in {\cal U}$, 
\begin{align*}
x  \geq 
\inf_{\sigma \in \mathcal{T}} \sup_{\tau \in \mathcal{T}}
\mathcal{E}_{0, \tau \wedge \sigma}^{g^{\alpha}} [ I(\tau, \sigma)].
\end{align*}
By taking the supremum over $\alpha \in {\cal U}$ in this inequality, we obtain
\begin{align*}
x  \geq 
\sup_ { \alpha \in {\cal U} }\inf_{\sigma \in \mathcal{T}} \sup_{\tau \in \mathcal{T}}
\mathcal{E}_{0, \tau \wedge \sigma}^{g^{\alpha}} [ I(\tau, \sigma)]={\bf v_0},
\end{align*}
where the last equality follows from the fact that ${\bf v_0}$ is equal to the value function of the {\em mixed generalized Dynkin game} \eqref{q2}. By taking the infimum over $x \in \mathcal{H}^r$, 
we obtain ${\bf u_0} \geq {\bf v_0}=Y_0$. Since $Y_0 \geq {\bf u_0}$, we thus get $Y_0 ={\bf u_0}$. Since $( \sigma^*,\varphi^*) \in {\cal S}^r (Y_0)$, we derive that $( \sigma^*,\varphi^*) \in {\cal S}^r ({\bf u_0})$.
The last assertion of the theorem follows from Theorem \ref{cha}.
\end{proof}

When $\zeta$ is only RCLL, by using similar arguments to those used in the above proof and in the proof of Theorem \ref{epsil}, one can show the following result. 
\begin{theorem}\label{srp2} [seller's robust price and $\varepsilon$-super-hedge]  Suppose that the process $\zeta$ and $\xi$ are only RCLL. 
 The {\em seller's robust price}  of the game option is equal the dual value function, that is 
${\bf u_0} = {\bf v_0}.$ We also have ${\bf u_0}=Y_0,$ where
$(Y, Z ,K, A, A')$ is the solution of the DRBSDE associated with  
driver ${\bf G}$  defined by \eqref{d} and barriers $\xi$ and $\zeta$.

 Moreover, the infimum in \eqref{q} is not necessarily attained. 
For each $\varepsilon>0$, let 
$\sigma_{\varepsilon}:= \inf \{t \geq 0: \,\,\, Y_t \geq \zeta_t-\varepsilon \}.$  
 The pair  
$(\sigma_{\varepsilon}, \varphi^*)$, where  $\varphi^*:=\Phi(Z,K)$,
is  an 
 $\varepsilon$-{\em robust super-hedge}  for the seller, in the sense that
 \begin{equation*}
 V_t^{\alpha ,u_0, \varphi^*} \geq \xi_t, \,\,0 \leq t \leq \sigma_{\varepsilon}\,\,\, {\rm a.s.} \quad {\rm and }\quad
V^{\alpha ,u_0, \varphi^*}_{\sigma_{\varepsilon}}  \geq   \zeta_{\sigma_{\varepsilon}}- 
\varepsilon  \,\,\,{\rm a.s.}\quad \forall \alpha \in {\cal U}.
\end{equation*}
%For each $\varepsilon>0$, let 
%$$\sigma_{\varepsilon}:= \inf \{t \geq 0: \,\,\, Y_t \geq \zeta_t-\varepsilon \}.$$ 
% Let us consider the risky assets strategy $\varphi^*:=\Phi(Z,K)$ (defined by \eqref{stbis}). We have
% \begin{equation}\label{case}
% V_t^{Y_0, \varphi^*} \geq \xi_t, \,\,0 \leq t \leq \sigma_{\varepsilon}\,\,\, {\rm a.s.} \quad {\rm and }\quad
%V^{Y_0, \varphi^*}_{\sigma_{\varepsilon}}  \geq   \zeta_{\sigma_{\varepsilon}}- 
%\varepsilon  \,\,\, {\rm a.s.}
%\end{equation}
%In other terms, the pair  
%$(\sigma_{\varepsilon}, \varphi^*)$ 
%is an 
% $\varepsilon$-super-hedge  for the initial capital amount $Y_0$. 
 \end{theorem}

%\begin{proposition}

%\begin{remark}
%Contrary to the case when there is no ambiguity, $\tau_{\varepsilon}:= \inf \{t \geq 0: \,\,\, Y_t \leq \xi_t+\varepsilon \}$ (resp. $\tau^*:= \inf \{t \geq 0: \,\,\, Y_t = \xi_t \}$) can no longer be interpreted as 
%an $\varepsilon$-{\em rational} (resp. {\em rational}) stopping time for the buyer.

%\end{remark}

We will now show that the infimum over $\sigma$ and the supremum over $\alpha$ can be interchanged in the expression of the dual value function ${\bf v_0}$ (see \eqref{q2}), which, since ${\bf u_0}={\bf v_0}$, can be written as follows.
\begin{proposition} \label{interversion} The {\em seller's robust price} ${\bf u_0}$ of the game option satisfies:
 \begin{equation}\label{inter}
 {\bf u_0}= \sup_{\alpha \in \mathcal{U}} \inf_{\sigma \in \mathcal{T}} \sup_{\tau \in \mathcal{T}}  \mathcal{E}^{g^{\alpha}}_{0, \tau \wedge \sigma}[I(\tau,\sigma)]= \inf_{\sigma \in \mathcal{T} }\sup_{\alpha \in \mathcal{U}}  \sup_{\tau \in \mathcal{T}}
 \cal{E}^{g^\alpha}_{0,\tau \wedge \sigma}[I(\tau, \sigma)].
 \end{equation}
\end{proposition}

\begin{proof} 
The first equality in \eqref{inter} holds by the above theorem. Let us prove the second one.
%From \eqref{q2}, we have to prove the following equality:
%
%\begin{align}
%\inf_{\sigma \in \mathcal{T}}   \sup_{\alpha \in \mathcal{U}} \sup_{\tau \in \mathcal{T}} \mathcal{E}_{0,\tau \wedge \sigma}^g[I(\tau,\sigma)]=\sup_{\alpha \in \mathcal{U}} \inf_{\sigma \in \mathcal{T}} \sup_{\tau \in \mathcal{T}}  \mathcal{E}^g_{0, \tau \wedge \sigma}[I(\tau,\sigma)]
%%= \sup_{\alpha \in \mathcal{U}} \sup_{\tau \in \mathcal{T}} \inf_{\sigma \in \mathcal{T}}  \mathcal{E}_{0, %\tau \wedge \sigma}^g[I(\tau,\sigma)].
%\end{align}
By the above theorem, we have ${\bf u_0}=Y_0$, where $(Y, Z ,K, A, A')$ is the solution of the DRBSDE associated with  
driver ${\bf G}$  defined by \eqref{d} and barriers $\xi$ and $\zeta$. To obtain the desired result, it is thus sufficient to prove that 
\begin{equation}\label{ud}
Y_0= \inf_{\sigma \in \mathcal{T} }\sup_{\alpha \in \mathcal{U}}  \sup_{\tau \in \mathcal{T}}
 \cal{E}^{g^\alpha}_{0,\tau \wedge \sigma}[I(\tau, \sigma)].
 \end{equation}
 % ICI
Since by definition \eqref{d}, ${\bf G}= \sup_{\alpha \in U} g(\cdot, \alpha)$, by using similar arguments to those used in the proof of Theorem \ref{cha} (in particular a measurable selection theorem), one can show that the solution of the BSDE associated with driver ${\bf G}$
% JUSQU'ICI
 and terminal condition $I(\tau,\sigma)$ is equal to the supremum over $\alpha$ of  the solutions of the BSDEs associated with drivers $g(\cdot, \alpha)$ and the same terminal condition, that is
\begin{equation}\label{uu}
\mathcal{E}_{0,\tau\wedge \sigma}^{\bf G}[I(\tau,\sigma)]= \sup_{\alpha \in \mathcal{U}} \mathcal{E}_{0,\tau\wedge \sigma}^{g^{\alpha}}[I(\tau,\sigma)].
\end{equation}
On the other hand, applying Proposition \ref{fairpricegame} to the {\em generalized
 Dynkin game} associated with driver ${\bf G}$, we obtain the equality
\begin{equation}\label{ddy}
Y_0= \inf_{\sigma \in \mathcal{T}} \sup_{\tau \in \mathcal{T}}  \mathcal{E}_{0,\tau\wedge \sigma}^{\bf G}[I(\tau,\sigma)].
\end{equation}
Combining  \eqref{uu} and  \eqref{ddy}, we obtain the desired equality \eqref{ud}.
 \end{proof}
\subsection{Application to the case of ambiguity on the default probability} \label{example}
% on mathematical issues, concerning in particular the nonlinear pricing systems 
%associated to the market models ${\cal M}_{\alpha}$, $\alpha \in {\cal U}$.
We consider a family of a priori probability measures parametrized by $\alpha \in {\cal U}$.
More precisely, for each $\alpha \in {\cal U}$, let $Q^{\alpha }$ be the probability measure equivalent to $P$, which admits $Z^{\alpha }_T$ as density with respect to $P$, where $(Z_t^{\alpha })$ is the solution of the following SDE:
$$
dZ_t^{\alpha }=Z_t^{\alpha } \nu(t,{\alpha}_t)dM_t; \quad Z^{\alpha }_0=1,
$$
where $\nu: (\omega,t,\alpha) \mapsto \nu(t, \omega,\alpha)$ is  a bounded ${\cal P}\otimes {\cal B}(U) $-measurable function defined on $\Omega  \times [0,T] \times U $ 
with $\nu (t,\alpha)> C_1   >-1$.

By Girsanov's theorem, we derive that under $Q^{\alpha }$, $W$ is a $\mathbb G$-Brownian motion and $M^{\alpha}_t:= N_t- \int_0^t \lambda_s (1+ \nu(s,{\alpha}_s))ds$ is a $\mathbb G$-martingale.  Hence, under $Q^{\alpha }$, the $\mathbb G$-default intensity is equal to 
$\lambda_t (1+ \nu(t,{\alpha}_t))$.
The process $\nu(t,{\alpha}_t)$ represents the {\em uncertainty on the default intensity}.

To each $\alpha \in {\cal U}$, corresponds a market model ${\cal M}_{\alpha}$ associated with the a priori probability measure $Q^{\alpha }$. In the market  ${\cal M}_{\alpha}$, the dynamics of the wealth process $V^{\alpha ,x, \varphi}$ associated with an initial wealth $x$ and  a risky assets stategy $\varphi$ $\in$ ${\mathbb H}^2 
\times {\mathbb H}^2_{\lambda}$ 
are supposed to satisfy 
 \begin{equation}\label{richessea}
-dV^{\alpha ,x, \varphi}_t= f(t,V^{\alpha ,x, \varphi}_t,  \varphi_t ' \sigma_t ,- \varphi_t^{2} , \alpha_t) dt -  \varphi_t ' \sigma_t dW_t +\varphi_t^{2}dM^{\alpha}_t\,;\,\,\,
V^{\alpha ,x, \varphi}_0=x, 
\end{equation}
%$$-dV_t= f(t,V_t, Z_t,K_t ) dt -  Z_t dW_t - K_t dM^{\alpha}_t.$$
where $f: (t, \omega, y,z, k,\alpha) \mapsto  f(t, \omega, y,z, k, \alpha)$ is a map  supposed to be {\em uniformly $\lambda$- admissible}  with respect to $(y,z, k)$, satisfying \eqref{rrrbis} with 
 $\theta^{t,y,z,k_1, k_2}> (-1-C_1) \vee (-1)$ and  $ \sup_{\alpha \in U}  | f(  t, .,0,0,0, \alpha)|  \in \H^{p}$, for some $p > 2$. 
 For example, $f$ can be given as in \eqref{perfectlineaire} in the case of a perfect market,
 or as in Examples \ref{eximp} of market imperfections, with coefficients which may depend on $\alpha$.
%\footnote{We can also suppose that  $f$ depends on $\alpha$, 
%that is,  is of the form $f(t,{\alpha}_t,y,z,k )$, in order   to take into account some ambiguity on the coefficients of the model.}

%Recall that under $Q^{\alpha }$, $W$ is a Brownian motion and $M^{\alpha}_t:= M_t- \int_0^t \lambda_s (1+ \gamma(s,{\alpha}_s))ds$ is a martingale independant of $W$. 
By \cite[Proposition A.3]{DQS4}, there is a martingale representation theorem for ${\mathbb G}$-martingales under $Q^{\alpha}$ with respect to $W$ and $M^{\alpha}$.
Let $\xi \in {L}^p({\cal G_T})$, where $p>2$. By  \cite[Proposition 2.11]{DQS4}, the density 
$Z^{\alpha}_T$ of $Q^{\alpha}$ with respect to $P$ belongs to $L^q$
for all $q\geq 2$. 
Let $p' \in ]2, p[$.
Applying H\"{o}lder's inequality, we derive that $E_{Q^{\alpha}}(\xi^{p'}) < + \infty$. Similarly, since by assumption $f(  t, 0,0,0, \alpha_t)  \in \H^{p}$, we derive that 
$   f(  t, 0,0,0, \alpha_t)   \in \H^{p'}_{Q^{\alpha}}$.
By 
\cite[Corollary A.4]{DQS4}, 
there exists an unique solution  $(X^{\alpha}, Z^{\alpha}, K^{\alpha})$ in $ \mathcal{S}_{Q^{\alpha}}^{p'} \times {\mathbb H}_{Q^{\alpha}} ^{p'} \times  {\mathbb H}^{p'}_{Q^{\alpha},\lambda}$ of the following $Q^{\alpha }$-BSDE:
\begin{equation}\label{BSDEalpha}
-dX^{\alpha}_t = f(t,X^{\alpha}_t, Z^{\alpha}_t,K^{\alpha}_t , \alpha_t) dt  -  Z^{\alpha}_t dW_t - K^{\alpha}_t dM^{\alpha}_t; \quad
X^{\alpha}_T=\xi.
\end{equation}

%We are given a family of probability measures parametrized by $\alpha \in {\cal U}$.
%More precisely, for each $\alpha \in {\cal U}$, let $Q^{\alpha }$ be the probability measure which admits $Z^{\alpha }_T$ as density with respect to $P$, where $(Z_t^{\alpha })$ is the solution of the following SDE:
%$$
%dZ_t^{\alpha }=Z_t^{\alpha } \gamma(t,{\alpha}_t)dM_t; \quad Z^{\alpha }_0=1,
%$$
%where $\gamma$ are bounded with $\gamma (t,\alpha)> C_1   >-1$.
%
%Under $Q^{\alpha }$, $W$ is a Brownian motion and $M^{\alpha}_t:= M_t- \int_0^t \lambda_s (1+ \gamma(s,{\alpha}_s))ds$ is a martingale independant of $W$.
%The process $\gamma(\cdot,{\alpha}_\cdot)$ represents the {\em uncertainty on the default intensity}.
\noindent As in the previous section, to simplify notation, for each $\alpha \in \mathcal{U}$, we denote by $f^{\alpha}$ the driver  
 $%\begin{equation}\label{gg}
 f^{\alpha}(t, y,z,k)= f(t, y,z,k,  \alpha_t)$. The {\em nonlinear price system} in the market model ${\cal M}_{\alpha}$, denoted by $\mathcal{E}^{f^{\alpha}}_{Q^{\alpha} }$,   is thus the $f^{\alpha}$-evaluation under  the 
{\em a priori probability} measure $Q^{\alpha }$, defined on $L^{p'}$.
The robust super-hedges are defined as in Definition \ref{definitionrobust} and the seller's robust price ${\bf u_0}$ is defined by 
\eqref{robustprice}.
Since $M^{\alpha}_t= M_t- \int_0^t \lambda_s  \nu(s,{\alpha}_s)ds$, the dynamics \eqref{richessea} of the wealth process $V^{\alpha ,x, \varphi}$ in the market model ${\cal M}_{\alpha}$ can be written as follows: 
 \begin{equation*}
-dV^{\alpha ,x, \varphi}_t= - \lambda_t \nu(t,\alpha_t)\varphi_t^{2} dt+ f(t,V^{\alpha ,x, \varphi}_t,  \varphi_t \sigma_t ,- \varphi_t^{2} , \alpha_t) dt -  \varphi_t \sigma_t dW_t +\varphi_t^{2}dM_t.
\end{equation*}
 This example thus corresponds to the model with ambiguity defined in Section \ref{mamo}  with $g(\cdot, \alpha)$ defined  by 
% \footnote{
%% %
%% Note that when $f$ corresponds to \eqref{perfectlineaire} as in the case of a perfect market, then $g(\cdot, \alpha)$ is given by 
%% $g(t,y,z,k,  \alpha)=  \nu(t,\alpha) \lambda_t k- r_t y - (z + \sigma^2_t k\,{\bf 1}_{\{t \leq \vartheta \} }) \theta^1_t -  \theta^2_t \lambda_t k $.
%% %
% }
% \begin{equation} \label{GG}
 $$g(t, \omega,y,z,k,  \alpha):= \lambda_t(\omega) \nu(t,\omega,\alpha)k+f(t,\omega, y,z,k, \alpha).$$
% \end{equation}
By the assumptions on $f$, the map $g$ satisfies the required conditions, in particular  inequality \eqref{rrrbis}.
% As in the previous section, to simplify notation, for each $\alpha \in \mathcal{U}$, we denote by $g^{\alpha}$ the driver  
% $%\begin{equation}\label{gg}
% g^{\alpha}(t, y,z,k)= g(t, y,z,k,  \alpha_t).
%$%  \end{equation}
Theorems \ref{srp} and \ref{srp2} as well as Proposition \ref{interversion} hold. In particular,  the {\em seller's robust price} ${\bf u_0}$ of the game option admits the following dual representation:
 \begin{equation}\label{inter2}
 {\bf u_0}= \sup_{\alpha \in \mathcal{U}} \inf_{\sigma \in \mathcal{T}} \sup_{\tau \in \mathcal{T}}  \mathcal{E}^{g^\alpha}_{0, \tau \wedge \sigma}[I(\tau,\sigma)]= \inf_{\sigma \in \mathcal{T} }\sup_{\alpha \in \mathcal{U}}  \sup_{\tau \in \mathcal{T}}
 \cal{E}^{g^\alpha}_{0,\tau \wedge \sigma}[I(\tau, \sigma)].
 \end{equation}
%Inutile (c'est ecrit avec G)Moreover, using the assumptions made on $f$ and the boundedness property of  $\gamma$, one can show that for each $\alpha$, $g^\alpha$  satisfies Assumption 
%\ref{Royer}. 
%
We now show that  for each $\alpha \in \mathcal{U}$, $ \mathcal{E}^{g^\alpha}$ is equal to the nonlinear price system $\mathcal{E}^{f^{\alpha} }_{Q^{\alpha} }$ relative to the market model ${\cal M}_{\alpha}$.
First, we have $(Z^{\alpha}_T)^{-1}$ $\in$ $L^{q}$ for all $q \geq 1$. Indeed,
The process $(Z_t^{\alpha })^{-1}$ satisfies the following $Q^{\alpha }$-SDE:
$
d(Z_t^{\alpha })^{-1}=-(Z_{t^-}^{\alpha })^{-1}\nu(t,{\alpha}_t)dM^{\alpha }_t,$ with $(Z_0^{\alpha })^{-1}=1
$.
By  \cite[Proposition 2.11]{DQS4}, $(Z^{\alpha}_T)^{-1}$ belongs to $L^{q'}_{Q^{\alpha }}$
for all $q'\geq 1$, which implies that $(Z^{\alpha}_T)^{-1}$ $\in$ $L^{q}$ for all $q \geq 1$.
Since $p' >2$, by H\"older's inequality, 
we derive that $(X^\alpha, Z^\alpha, K^\alpha)$ (solution of \eqref{BSDEalpha}) belongs to $ S^{2}\times \H^{2} \times \H_{ \lambda}^{2}$ and is thus the unique solution in $ S^{2}\times \H^{2} \times \H_{ \lambda}^{2}$ of the $P$-BSDE: 
\begin{equation*}
-dX^{\alpha}_t = g^{\alpha}(t,X^{\alpha}_t, Z^{\alpha}_t,K^{\alpha}_t ) dt  -  Z^{\alpha}_t dW_t - K^{\alpha}_t dM_t; \quad
X^{\alpha}_T=\xi.
\end{equation*}
Hence, for each maturity $S$ and  each payoff $\eta$ $\in$ $L^p({\cal G}_S)$, we have
$$\mathcal{E}^{f^{\alpha} }_{Q^{\alpha },\cdot,S}(\eta  )= \mathcal{E}^{g^{\alpha}}_{\cdot, S}(\eta),$$
which gives that $ \mathcal{E}^{g^\alpha}$ is equal to the nonlinear price system $\mathcal{E}^{f^{\alpha} }_{Q^{\alpha} }$ relative to the market model ${\cal M}_{\alpha}$.
Using this property together with equalities  \eqref{inter2} and Theorem \ref{srp2}, we derive the following result.
\begin{proposition}\label{exemple}
(Seller's robust price) 
%Suppose that $\zeta$ is left-lower semicontinuous along stopping times (and $\xi$ is only RCLL). 
 The {\em seller's robust price}  of the game option in this model admits the following dual representation: 
  \begin{equation} \label{uuu}
 {\bf u_0}= \sup_{\alpha \in \mathcal{U}} \inf_{\sigma \in \mathcal{T}} \sup_{\tau \in \mathcal{T}} \mathcal{E}^{f^{\alpha} }_{Q^{\alpha },0, \tau \wedge \sigma}[I(\tau,\sigma)]= \inf_{\sigma \in \mathcal{T} }\sup_{\alpha \in \mathcal{U}}  \sup_{\tau \in \mathcal{T}}
\mathcal{E}^{f^{\alpha} }_{Q^{\alpha },0, \tau \wedge \sigma}[I(\tau, \sigma)].
\end{equation}
Let ${\bf G}$ be the map defined for each $(t, \omega, z, k)$ by
\begin{equation}\label{ddd}
{\bf G}(t, \omega, y,z, k):=\sup_{\alpha \in U} \left(\lambda_t(\omega) \nu(t,\omega,\alpha)k+f(t,\omega, y,z,k, \alpha)\right).
\end{equation}
We have ${\bf u_0}=Y_0$, where  $Y$ is the solution of the $P$-DRBSDE associated with  
driver ${\bf G}$   and barriers $\xi$ and $\zeta$.
\end{proposition}

\section{Complementary results}\label{sec-comp}
\subsection{Pricing of European options from the buyer's point of view }\label{optionbuyer}
Let us consider the pricing and hedging problem of a European option with maturity $T$ and payoff $\xi \in L^2({\cal G}_T)$ from 
the buyer's point of view. Supposing the initial price of the option is $z$, he starts with the amount $-z$ at time $t=0$, and looks to find a risky-assets strategy $\tilde \varphi$ such that the payoff  that he receives at time $T$ allows him to recover the debt he incurred at time $t=0$ by buying the option, that is such that 
$V^{-z, \tilde \varphi}_T + \xi=0 \quad {\rm a.s.}\,$
%$$V^{-z, \tilde \varphi}_T + \xi=0 \quad {\rm a.s.}\,$$ 
or equivalently, $V^{-z, \tilde \varphi}_T =- \xi$ a.s.\,

The buyer's price of the option is thus equal to 
the opposite of the seller's price of the  option with payoff  $-\xi$, that is $-{\cal E}_{0,T}^{^{g}} (-\xi)=- \tilde X_0$, where
$( \tilde X,  \tilde Z, \tilde K)$ is the solution of the BSDE associated with driver $g$ and terminal condition $-\xi$. Let us specify the hedging strategy for the buyer. Suppose that the initial price of the option is $z:= - \tilde X_0$. The process $\tilde X$ is equal to  the value of the portfolio associated with initial value $-z= \tilde X_0$ and
strategy $\tilde \varphi $ $:= \Phi  
( \tilde Z,\tilde K)$ (where $\Phi$ is defined in Definition \ref{stbis})  that is
 $\tilde X= V^{\tilde X_0, \tilde \varphi}= V^{-z, \tilde \varphi}$. Hence, $V^{-z, \tilde \varphi}_T = \tilde X_T= -\xi$ a.s.\,, which yields that  $\tilde \varphi$ 
is the hedging risky-assets strategy for the buyer. Similarly, $-{\cal E}_{t,T}^{^{g}} (-\xi)=- \tilde X_t$ satisfies an analogous property at time $t$, and is is called the {\em hedging price for the buyer}
% {\em hedging price} 
 at time $t$.

% In conclusion, the buyer's price at time $t=0$ of the European option in the market model ${\cal M}^g$ is equal to 
%$$ \tilde {\cal E}_{0,T}^{^{g}} (\xi):=
%-{\cal E}_{0,T}^{^{g}} (-\xi)= - X_0(T,- \xi),$$ and the strategy $\tilde \varphi$ 
%is the hedging strategy for the buyer. 

%his price at time $t=0$ in the market model ${\cal M}^g$ would be equal to 
%$ \tilde {\cal E}_{0,T}^{^{g}} (\xi):=
%-{\cal E}_{0,T}^{^{g}} (-\xi)= - X_0(T,- \xi)$, where $(X(T,- \xi), Z(T,- \xi),K(T,- \xi))$ is the solution of the BSDE associated with driver $g$ and terminal condition $-\xi$. Indeed, setting $z:= -X_0(T,- \xi)$ and $\tilde \varphi $ $= \Phi  
%( Z(T,- \xi),K(T,- \xi))$,  %(where $\Phi$ is defined in Definition \ref{stbis}), 
%we have $V^{-z, \tilde \varphi}_T + \xi=0$ a.s. {\bf Hence, supposing the initial price of the option is $z$, he starts with the amount $-z=\tilde X_0$ at time $t=0$, and, following  the risky-assets strategy $\tilde \varphi$, the payoff  that he receives at time $T$ allows him to recover the debt he incurred at time $t=0$ by buying the option. The buyer's price of the option with payoff $\xi$ is thus equal to 
%the opposite of the seller's price of the  option with payoff  $-\xi$ and the strategy $\tilde \varphi$ 
%is the hedging strategy for the buyer. 
This leads to the {\em nonlinear pricing system $\tilde {\cal E}^{^{g}}$ relative to the buyer} in the market ${\cal M}^g$
 defined for each $(S, \xi) \in [0,T]\times L^2({\cal G}_S)$ by 
 \begin{equation}\label{tildeE}
 \tilde {\cal E}^{^{g}}_{\cdot, S}(\xi):=
-{\cal E}^{^{g}}_{\cdot, S} (-\xi).
\end{equation}
%
%
%\begin{remark}
%Note that when $g(t,0,0,0)\geq 0$ a.s.\,, by the comparison theorem for BSDEs with default jump (see Theorem 2.17 in \cite{DQS4}), the nonlinear pricing systems are nonnegative, that is, 
%if $\xi \geq 0$, then ${\cal E}^{^{g}}_{\cdot, S} (\xi)\geq 0$. 
%\end{remark}

%\begin{remark}
% Note that the hedging price $x:={\cal E}_{0,T}^{^{g}} (\xi)$ is clearly an {\em upper bound} of the possible prices for the European option. 
%Indeed, no  rational agent would pay more than $x$
% since there is a cheaper way to achieve at least the same payoff. 
%  Indeed,  by investing the amount $x+  \varepsilon$ and  following the  strategy $\varphi $, 
% he will make a gain $V^{x+  \varepsilon, \varphi}_T  > V^{x, \varphi}_{T} (\geq \xi)$ by a  strict comparison property for deterministic differential equations (see e.g. Lemma \ref{classique}).
% \end{remark}
%
\begin{remark} \label{perfectegal}
When $g(t,0,0,0) =0$, then $\tilde {\cal E}^{^{g}}_{\cdot, S}(0)=0$. Moreover, by the comparison theorem for BSDEs with default, if $\xi \geq 0$, then $\tilde {\cal E}^{^{g}}_{\cdot, S}(\xi)\geq 0$.

Note that $\tilde {\cal E}^{^{g}}_{\cdot, S}(\xi)$ is equal to the solution of the BSDE with driver $-g(t,-y,-z,-k)$ and terminal condition $\xi$. Hence, if we suppose that 
$-g(t,-y,-z,-k) \leq g(t,y,z,k)$ (which is satisfied if, for example, $g$ is convex with respect to $(y,z,k)$), then, by the comparison theorem for BSDEs, we have  $\tilde {\cal E}^{^{g}}_{\cdot, S}(\xi)= -{\cal E}^{^{g}}_{\cdot, S} (-\xi) \leq  {\cal E}^{^{g}}_{\cdot, S}(\xi)$  for each $(S, \xi) \in [0,T]\times L^2({\cal G}_S)$.
\footnote{Note that a price functional $p$ generally satisfies $-p (-\xi) \leq  p(\xi)$(see e.g. \cite{Jouini} Section 2).}

Moreover, when $-g(t,-y,-z,-k) = g(t,y,z,k)$ (which is satisfied if, for example, $g$ is linear with respect to $(y,z,k)$, as in the perfect market case), we have $\tilde {\cal E}^{^{g}}=  {\cal E}^{^{g}}$. 
%Indeed, for each $(S, \xi) \in [0,T]\times L^2({\cal G}_S)$, $\tilde {\cal E}^{^{g}}_{\cdot, S}(\xi)=
%-{\cal E}^{^{g}}_{\cdot, S} (-\xi)= {\cal E}^{^{g}}_{\cdot, S} (\xi)$.
%
\end{remark}

 \subsection{Pricing of the game option from the buyer's point of view }
 In this section, we consider  the point of view of the buyer of the game option. Supposing the initial price of the game option is $z$, he starts with the amount $-z$ at time $t=0$, and looks to find a {\em super-hedge}, that is
an exercise time $\tau$  and a risky-assets strategy $\varphi$, such that the payoff  that he receives allows him to recover the debt he incurred at time $t=0$ by buying the game option, no matter the cancellation time chosen by the seller. This notion of super-hedge for the buyer can be defined more precisely as follows.

 \begin{definition}
  A {\em buyer's super-hedge} against the game option with initial price $z\in {\mathbb R}$ is a pair $(\tau, \varphi)$ of a stopping time $\tau \in {\cal T}$ and a risky-assets strategy $ \varphi$ $\in$  ${\mathbb H}^2 \times  {\mathbb H}^2_{\lambda}$ such that 
  \begin{equation}\label{condB}
V^{-z, \varphi}_{t } \geq -\zeta_t , \; 0\leq t < \tau \; \text{ a.s. and } V^{-z, \varphi}_{\tau } \geq - \xi_{\tau}
 \text{ a.s.}
\end{equation} 
We denote by $\mathcal{B}_{\xi, \zeta}(z)$ the set of all  {\em buyer's super-hedges} against the game option with payoffs $(\xi, \zeta)$
associated with initial price $z\in {\mathbb R}$. 

The {\em buyer's  price} of the game option in the market model ${\cal M}^g$, denoted by $\tilde u_0$,  is defined as the supremum of the initial prices which allow the buyer to be super-hedged,
 that is \footnote{
We have $(0,0) \in {\cal B}_{\xi, \zeta}(\xi_0)$. Hence, $\tilde u_0\geq \xi_0$.
Moreover, similarly to Remark \ref{positive2}, if $g(t,0,0,0)=0$ and $\xi_0\geq 0$, then 
$\tilde u_0= \sup\{z \geq 0, \,\,\, \exists (\tau, \varphi) \in\mathcal{B}_{\xi, \zeta}(z)\}$.}
\begin{equation}\label{rbp}
\tilde u_0:= \sup\{z \in \mathbb{R}, \,\,\, \exists (\tau, \varphi) \in\mathcal{B}_{\xi, \zeta}(z)\}.
\end{equation}

\end{definition}
%We denote by $\mathcal{B}_{\xi, \zeta}(z)$ the set of all  buyer's super-hedges against the game option with payoffs $(\xi, \zeta)$
%associated with initial price $z\in {\mathbb R}$. 
%\begin{remark}\label{rem}
%Condition \eqref{condB} is equivalent to $V^{-z, \varphi}_{t \wedge \tau} +I(\tau,t) \geq 0$, $0 \leq t \leq T$ a.s.\, (utile?)
The first inequality of \eqref{condB} also holds at time $t= \tau$  
because $\xi \leq \zeta$.
%, it is also equivalent to 
%$$V^{-z, \varphi}_{t } \geq -\zeta_t , \; 0\leq t \leq \tau \; \text{ a.s. and } V^{-z, \varphi}_{\tau } \geq - \xi_{\tau}.$$
It follows that ${\cal B}_{\xi, \zeta}(z)= {\cal S}_{-\zeta, -\xi}(-z)$, where ${\cal S}_{-\zeta, -\xi}(-z)$ is the set of 
{\em seller's super-hedges} against the game option with payoffs $(-\zeta, -\xi)$
associated with initial capital $-z$.\\
Hence,
$- \tilde u_0= \inf \{x \in \mathbb{R}, \,\,\, \exists (\tau, \varphi) \in{\cal S}_{-\zeta, -\xi}(x)\}.$
We thus have:
\begin{theorem}\label{bs}
The {\em buyer's price} of the game option with payoffs $(\xi,\zeta)$ is equal to 
the opposite of the {\em seller's price} of the game option with payoffs  $(-\zeta, -\xi)$. 
%\footnote{
%% A GARDER 
%{\bf Theorem \ref{bs} still holds at time $t$. More precisely, the buyer's price at time $t$, defined as the essential supremum of the initial prices (at time $t$) which enable the buyer to be super-hedged, is equal
%the opposite of the seller's price at time $t$ of the game option with payoffs  $(-\zeta, -\xi)$, defined in Remark \ref{Important}.
%}
%%
%}
%
\end{theorem}

The previous results (Theorem \ref{superhedging} and Theorem \ref{epsil}) can thus be applied.
In particular, we have the following dual formulation of the buyer's price:
\begin{equation} \label{qbb}
\tilde u_0 =     \sup_{\tau \in \mathcal{T}} \inf_{\sigma \in \mathcal{T} } 
 \tilde { \cal{E}}^{g}_{0,\tau \wedge \sigma} [I(\tau, \sigma)] =\inf_{\sigma \in \mathcal{T} }  \sup_{\tau \in \mathcal{T}} \tilde { \cal{E}}^{g}_{0,\tau \wedge \sigma}[I(\tau, \sigma)], 
\end{equation} where 
$ \tilde  {\cal{E}}^{g}_{0,\tau \wedge \sigma} [I(\tau, \sigma)] = -\cal{E}^{g}_{0,\tau \wedge \sigma} 
 [- I(\tau, \sigma)] $. 
The quantity $ \tilde  {\cal{E}}^{g}_{0,\tau \wedge \sigma} [I(\tau, \sigma)] $ corresponds to the buyer's price  of the European option with payoff $I(\tau, \sigma)$ and terminal time $\tau \wedge \sigma$
%({\bf A SUPPRIMER}: see Remark \ref{buyereuro}). 
(see \eqref{tildeE}).
%From this result together with , 
%one can derive properties satisfied by the buyer's price analogous to those satisfied by the seller's price.
\begin{remark}
In the special case of a perfect market, the dynamics of the wealth process $X$ are linear with respect to $(X, \varphi)$, which implies that 
the {\em buyer's price} $\tilde u_0$ is equal to the {\em seller's price} $ u_0$ (and $\tilde {\cal E}^{^{g}}=  {\cal E}^{^{g}}$, as seen in Remark \ref{perfectegal}).
\end{remark}

Let  $(\Tilde{Y}, \Tilde{Z}, \Tilde{K}, \Tilde{A}, \Tilde{A}')$ be the solution of the DRBSDE associated with  
driver $g$  and barriers $(-\zeta,-\xi)$. By  Theorem \ref{epsil}, the {\em buyer's price} is equal to 
the opposite of the solution, that is, $ \tilde  u_0=-\Tilde{Y}_0.$ 

Moreover, by Theorem \ref{superhedging}, when $\xi$  is left-u.s.c. along stopping times (but not necessarily $-\zeta$), the pair $(\tilde \tau, \tilde \varphi)$, where
$\tilde \tau:= \inf \{t \geq 0: \,\,\, -\tilde Y_t = \xi_t\}$
and $\tilde \varphi $ $:= \Phi  
( \tilde Z, \tilde K)$,  is a {\em buyer's super-hedge}.

\paragraph{Buyer's robust  price of the game option in the case with ambiguity.}

In this paragraph, we consider the market model with ambiguity described in Section \ref{mamo}.
 \begin{definition}
 A {\em buyer's robust super-hedge} against the game option with initial price $z\in {\mathbb R}$ is a pair $(\tau, \varphi)$ of a stopping time $\tau \in {\cal T}$ and a strategy $ \varphi$ $\in$  ${\mathbb H}^2 \times  {\mathbb H}^2_{\lambda}$ such that 
% For a given initial wealth $z \in {\mathbb R}$, a {\em buyer robust super-hedge} against the game option is a pair $(\tau, \varphi)$ of a stopping time $\tau \in {\cal T}$ and a risky-assets strategy $ \varphi$ $\in$  ${\mathbb H}^2 \times  {\mathbb H}^2_{\lambda}$ such that  
  \begin{equation}\label{condBa}
V^{\alpha ,z, \varphi}_{t } \geq -\zeta_t , \; 0\leq t < \tau \; \text{ a.s. and } V^{\alpha ,z, \varphi}_{\tau } \geq - \xi_{\tau}
 \text{ a.s.}\,, \quad \forall \alpha \in {\cal U}.
\end{equation} 
We denote by $\mathcal{B}^r_{\xi, \zeta}(z)$ the set of all  {\em buyer's robust super-hedges} against the game option with payoffs $(\xi, \zeta)$
associated with initial price $z\in {\mathbb R}$.

The {\em buyer's  robust price} of the game option is defined as the supremum of the initial prices which allow the buyer to 
construct a {\em robust superhedge}, that is
\begin{equation}\label{robustb}
{\bf \tilde  u_0}:= \sup\{z \in \mathbb{R}, \,\,\, \exists (\tau, \varphi) \in \mathcal{B}^r_{\xi, \zeta}(z)\}.
\end{equation}

\end{definition}
 Since $\xi \leq \zeta$, condition \eqref{condBa} is equivalent to 
$$V^{-z, \varphi}_{t } \geq -\zeta_t , \; 0\leq t \leq \tau \; \text{ a.s. and } V^{-z, \varphi}_{\tau } \geq - \xi_{\tau}
 \text{ a.s.}\,, \quad \forall \alpha \in {\cal U}.$$
It follows that ${\cal B}^r_{\xi, \zeta}(z)= {\cal S}^r _{-\zeta, -\xi}(-z)$, where ${\cal S}^r _{-\zeta, -\xi}(-z)$ is the set of 
 seller's robust super-hedges against the game option with payoffs $(-\zeta, -\xi)$
associated with initial capital $-z$. We thus have
\begin{theorem}
The {\em buyer's robust price} of the game option with payoffs $(\xi,\zeta)$ is equal to 
the opposite of the {\em seller's robust price} of the game option with payoffs  $(-\zeta, -\xi)$. 
\end{theorem}
The previous results (Theorem \ref{srp} and \ref{srp2}) can thus be applied. In particular, we have the following dual formulation of the {\em buyer's robust price}:
\begin{equation} \label{qter}
{\bf \tilde  u_0} = \inf_{\alpha \in \mathcal{U}}    \sup_{\tau \in \mathcal{T}} \inf_{\sigma \in \mathcal{T} } \tilde{\cal{E}}^{g^{\alpha}}_{0,\tau \wedge \sigma}[I(\tau, \sigma)]
= \inf_{\alpha \in \mathcal{U}} \inf_{\sigma \in \mathcal{T} } \sup_{\tau \in \mathcal{T}} \tilde{\cal{E}}^{g^{\alpha}}_{0,\tau \wedge \sigma}[I(\tau, \sigma)], 
\end{equation} where $\tilde{\cal{E}}^{g^{\alpha}}_{0,\tau \wedge \sigma}[I(\tau, \sigma)] = 
-\cal{E}^{g^{\alpha}}_{0,\tau \wedge \sigma}[-I(\tau, \sigma)]$. 
Using \eqref{qbb}, we derive that the {\em buyer's robust price }${\bf \tilde  u_0}$ is equal to the infimum over $\alpha \in \mathcal{U}$ of the 
buyer's  prices in ${\cal M}_{\alpha}$. 
\begin{remark} 
By Proposition \ref{interversion}, we derive that 
$$
{\bf \tilde  u_0} = 
\sup_{\tau \in \mathcal{T}} \inf_{\alpha \in \mathcal{U}}    \inf_{\sigma \in \mathcal{T} } \tilde{\cal{E}}^{g^{\alpha}}_{0,\tau \wedge \sigma}[I(\tau, \sigma)].
$$
Note that 
%({\bf A SUPPRIMER: by Remark \ref{buyereuro}}), 
for each $\alpha \in \mathcal{U}$, the quantity $\tilde{\cal{E}}^{g^{\alpha}}_{0,\tau \wedge \sigma}[I(\tau, \sigma)]$ is the buyer's price in the market model ${\cal M}_{\alpha}$ of the European option with payoff $I(\tau, \sigma)$ and terminal time $\tau \wedge \sigma$ (see \eqref{tildeE}).
\end{remark}
Let  $(\Tilde{Y}, \Tilde{Z}, \Tilde{K}, \Tilde{A}, \Tilde{A}')$ be the solution of the DRBSDE associated with  
driver ${\bf G}$  defined by \eqref{d} and barriers $(-\zeta,-\xi)$. By Theorem \ref{srp2}, 
the {\em buyer's robust price} of the game option is equal to $-\Tilde{Y}_0$, that is,
 ${\bf \tilde  u_0}=-\Tilde{Y}_0.$
 Moreover, by Theorem \ref{srp}, when $\xi$  is left-u.s.c. along stopping times (but not necessarily $-\zeta$), the pair $(\tilde \tau, \tilde \varphi)$, where $\tilde \tau:= \inf \{t \geq 0: \,\,\, -\tilde Y_t = \xi_t\}$ and $\tilde \varphi $ $:= \Phi  
( \tilde Z, \tilde K)$,  is a {\em buyer's robust super-hedge} of the game option.
 \subsection{Seller's price and buyer's price {\em processes} of the game option}
 We can define the seller's price of the game option at each stopping time $S \in {\cal T}$. More precisely, for each wealth  $X\in L^2( {\cal F}_S) $ (at initial time $S$), an {\em $S$-super-hedge} against the game option is a pair $(\sigma, \varphi)$ of a stopping time $\sigma \in {\cal T}_S$ and a portfolio strategy $ \varphi$ $\in$  ${\mathbb H}^2 \times  {\mathbb H}^2_{\lambda}$ such that
$
V^{S,X, \varphi}_{t } \geq \xi_t,$  $S\leq t \leq \sigma$ a.s. and $ V^{S,X, \varphi}_{\sigma } \geq \zeta_{\sigma}$ a.s.\,,
where $V^{S,X, \varphi}$ denotes the wealth process associated with initial time $S$ and initial condition $X$.
The {\em seller's price} at time $S$ is defined by 
$
 u(S):= {\rm ess} \inf \{X\in L^2( {\cal F}_S),\,\, \exists  (\sigma, \varphi) \in {\cal S}_S (X) \},
$ where ${\cal S} _S(X) $ is the  set of all $S$-super-hedges associated with initial wealth $X$. Using similar arguments to those used in 
the proof of Theorem \ref{epsil}, we obtain:
$$u(S)={\rm ess} \inf_{\sigma \in {\cal T}_S }\,{\rm ess} \sup_{\tau \in {\cal T}_S} \, \cal{E}^g_{S,\tau \wedge \sigma}(I(\tau, \sigma))=
{\rm ess} \sup_{\tau \in {\cal T}_S} \, {\rm ess} \inf_{\sigma \in {\cal T}_S} \, \cal{E}^g_{S,\tau \wedge \sigma}(I(\tau, \sigma)) =Y_S \quad {\rm a.s.}$$
where $(Y, Z, K, A,A')$ is the  solution of 
DRBSDE \eqref{DRBSDE}.\\
%The same property holds for the seller's robust price in the model with ambiguity. 
Similarly, we can define the  {\em buyer's price} at time $S$.

 \subsection{Game options with intermediate
dividends }
Suppose that a European option pays a terminal payoff $\xi$ at terminal time $S$ and an intermediate dividend, modeled by a nondecreasing  RCLL  adapted process 
$(D_t)$ with $D_0=0$. 
There exists an unique solution  $(X, Z, K)$  in $ \mathcal{S}^2 \times {\mathbb H}^2 \times  {\mathbb H}^2_{\lambda}$ of the following BSDE:
\begin{equation}\label{syste}
-dX_t = g(t,X_t, Z_t,K_t ) dt + dD_t -  Z_t dW_t - K_t dM_t; \quad
X_S=\xi.
\end{equation}
The process $X$ is  the wealth process associated with initial value $x= X_0$
 and strategy $\varphi $ $= \Phi  ( Z,K)$. Here, $dD_t$ represents the amount withdrawn from the portfolio between $t$ and $t + dt$ in order to pay the dividends to the buyer.
Hence, the amount $X_0$ allows the seller to be perfectly hedged against the option, in the sense that 
 it allows him/her to pay the intermediate dividends and the terminal payoff to the buyer, 
by investing the amount $X_0$ along the strategy  $\varphi $  in the market.
The  price for the seller  (at time $0$)  
of this option is thus given by $X_0$ and the associated hedging strategy is equal to $\varphi $.
Note that the driver of BSDE \eqref{syste} is given by the $\lambda$-{\em admissible}  
 ``generalized" driver $g(t,X_t, Z_t,K_t ) dt + dD_t$. 
 This  leads  to  the following {\em nonlinear pricing} system:\\
% , 
% first introduced in \cite{EQ96} in a Brownian framework (later called 
% {\em $g$-evaluation} in 
%\cite{Peng2004}) and denoted by ${\cal E}^g$.
For each $S\in [0,T]$, for each $\xi \in {L}^2({\cal G_S})$ 
and for each $D \in {\cal A}^2$, the associated 
{\em $g$-value} is defined by 
${\cal E}_{t,S}^{^{g,D}} (\xi):= X^D_t(S, \xi)$ for each $t \in [0,S]$. Note that ${\cal E}_{t,S}^{^{g,D}} (\xi)$ can be defined 
on the whole interval $[0,T]$ by setting ${\cal E}_{t,S}^{^{g,D}} (\xi):= {\cal E}_{t,T}^{^{g^S,D^S}} (\xi)$ for $t \geq S$, where 
$g^S(t,.):= g (t,.) {\bf 1}_{t \leq S}$ and $D_t^S := D_{t \wedge S}$. 
Some properties of this nonlinear pricing system are provided in \cite{DQS4}.\\
Concerning the pricing of the game option, the approach is the same, replacing the driver $g$ by the ``generalized" driver $g(\cdot ) dt + dD_t$, and ${\cal E}^{^{g}}$ by ${\cal E}^{^{g,D}}$.

\section{Appendix}

%Using the characterization of the solution of a doubly reflected BSDE as the (common) value 
%of a {\em generalized Dynkin game} (see Proposition \ref{fairpricegame}), 
We show the following estimates for DRBSDEs in our framework, with universal  constants. 
%(i.e. they only depend on the terminal time $T$ and the 
%common $\lambda$-constant $C$). 

\begin{proposition}[A priori estimate for DRBSDEs] \label{est}
Let $f^1$ be a $\lambda$-{\em admissible} driver with $\lambda$-constant $C$ and let  $f^2$ be a driver. 
Let $\xi$  and $\zeta$ be two adapted RCLL processes with $\zeta_T= \xi_T$ a.s.,   $\xi \in {\cal S}^2$,  $\zeta \in {\cal S}^2$, $\xi_t \leq \zeta_t$, $0 \leq t \leq T$ a.s.\,, and satisfying  Mokobodzki's condition.\\
 For $i=1,2$, let $(Y^i, Z^i ,K^i, A^i, A^{'i})$  be  a solution of the DRBSDE associated with  
terminal time $T$, driver $f^i$  and barriers $\xi$ and $\zeta$.
 Let $ \eta, \beta >0 $ be such that 
 $\beta \geq \frac{3}{\eta} +2C $ 
and $\eta \leq \frac{1}{C^2}$. \\
Let $\bar f(s): = f^1(s, Y^2_s, Z^2_s, K_s^2) - f^2(s, Y^2_s, Z^2_s, K_s^2)$.
For each $t \in [0,T]$, we then have
\begin{equation}\label{A26}
e^{\beta  t} (Y^1_s - Y^2_s)   ^2 \leq   \eta \,{\mathbb E}[ \int_t^T e^{\beta  s} \bar f(s)^2  ds \mid 
{\cal G}_t ] \;\; \text{ \rm a .s.}\, 
\end{equation}
Moreover, 
$\|\bar Y \|_\beta^2 \leq T  \eta
\|\bar f \|_\beta^2,$ and if $\eta < \frac{1}{C^2}$, we then have 
%\|\bar X \|_\beta^2 \leq T [e^{\beta T} E[(\xi^1 - \xi^2)^2] + \eta
%\|\bar f \|_\beta^2], \\
$\|\bar Z \|_\beta^2 + \|\bar K \|_{\lambda,\beta}^2
\leq \frac{\eta}{1 - \eta C^2}  \|\bar f \|_\beta^2.$

\end{proposition}
%\begin{remark}\label{coeff}
%Note  that the presence of the coefficient $\sqrt \lambda$ in the $\lambda$-{\em Lipschitz}  condition  \eqref{lip} (satisfied by $f^1$) is an important point to ensure that the universal constants in the above estimates do not depend on the ${\mathbb G}$-intensity $\lambda$, which is in particular interesting from a numerical point of view. This point is also important to ensure the existence and the uniqueness of the solution  (see the proof below).
%\end{remark}

\begin{proof}  For $s$ in $[0,T]$, denote $\bar Y_s := Y^1_s - Y^2_s, \,\,\, \bar Z_s := Z^1_s - Z^2_s$,  $\bar K_s := K^1_s - K^2_s $.
By It\^o's formula applied  to the semimartingale $e^{\beta s} \bar Y_s$ 
between $t$ and $T$, we get 
\begin{align}
e^{\beta t} \bar Y_t ^2  & + \beta \int_t^T e^{\beta s} \bar Y_s^2 ds +
 \int_t^T e^{\beta s} \bar Z_s^2 ds + \int_t^T e^{\beta s}  \bar K_s^2 \lambda_s ds
  + \sum_{0 <s \leq T} e^{\beta s }(\Delta A^1_s-\Delta A_s^2-\Delta A_s^{'1}+\Delta A_s^{'2})^{2} 
 \nonumber \\  
 &=  2 \int_t^T e^{\beta s} \bar Y_s (f^1 (s, Y^1_s, Z^1_s, K^1_s) - f^2 (s, Y^2_s, Z^2_s, K^2_s)) ds  \nonumber \\
 &\quad - 2 \int_t^T e^{\beta s} \bar Y_s \bar Z_s dW_s
  -  \int_t^T e^{\beta s} (2  \bar Y_{s^-} \bar K_s +  \bar K_s^2) d M_t \nonumber \\
&+ 2 \int^T_t e^{\beta s} \overline{Y}_{s^-}  dA^1_s - 2\int^T_t e^{\beta s} \overline{Y}_{s^-}  dA^2_s
- 2  \int^T_0 e^{\beta s} \overline{Y}_{s^-} \, dA'^1_s - \int^T_0 e^{\beta s} \overline{Y}_{s^-}\,  dA'^2_s .  \label{russ}
\end{align}
Now, we have 
$\overline{Y}_s dA^{1,c}_s = (Y^1_s - \xi_s)dA^{1,c}_s - (Y^2_s - \xi_s)dA^{1,c}_s = 
- (Y^2_s - \xi_s)dA^{1,c}_s\leq 0, $
and by symmetry, $\overline{Y}_s dA^{2,c}_s  \geq 0$. By similar arguments, we obtain $\overline{Y}_{s^-} \Delta A_{s}^{1,d}$ 
 $\leq 0 $,  $\,\overline{Y}_{s^-} \Delta A_{s}^{2,d}  \geq 0$, 
 $\, \overline{Y}_s dA^{' 1, c}_{s}  \geq 0 $, 
 $\overline{Y}_{s^-} \Delta A^{'1, d}_{s}\geq 0 $
and $\overline{Y}_{s^-} \Delta {{A'}^{2,d}}_s  \leq 0$.
Hence, the four last terms of the r.h.s. of \eqref{russ} are non positive.
%$$\int^T_t e^{\beta s} \overline{Y}_{s^-}  dA^1_s - \int^T_t e^{\beta s} \overline{Y}_{s^-}  dA^2_s \,\,\leq \, 0 
%\quad {\rm a.s.}$$
Taking the conditional expectation given ${\cal G}_t$, we obtain
\begin{align}
e^ {\beta t} & \bar Y_t^2 
 +  E \left[\beta \int_t^T e^{\beta s}  \bar Y_s^2 ds + \int_t^T e^{\beta s} ( \bar Z_s^2 + \bar K_s^2 \lambda_s) ds \mid {\cal G}_t \right] \nonumber \\
& \leq 2  E  \left[    \int_t^T e^{\beta s} \bar Y_s (f^1 (s, Y^1_s, Z^1_s, K^1_s) - f^2 (s, Y^2_s, Z^2_s, K^2_s)) ds   \mid {\cal G}_t\right].
\end{align}
Also, $|f^1(s,Y^1_s, Z^1_s, K^1_s) - f^2(s,Y^2_s,Z^2_s,K^2_s)|  \leq |f^1(s,Y^1_s, Z^1_s, K^1_s) - f^1(s,Y^2_s,Z^2_s,K^2_s)|  + |\bar f_s| $.\\
%where  
% $\bar f(s): = f^1(s, Y^2_s, Z^2_s, K_s^2) - f^2(s, Y^2_s, Z^2_s, K_s^2)$.\\
Using the $\lambda$-{\em admissibility} property of $f^1$, we derive that
$$
|f^1(s,Y^1_s, Z^1_s, K^1_s) - f^2(s,Y^2_s,Z^2_s,K^2_s)| 
%& \leq
% |f^1(s,Y^1_s, Z^1_s, K^1_s) - f^1(s,Y^2_s,Z^2_s,K^2_s)|  + |\bar f_s| \\
  \leq C  | \bar Y_s| + C|\bar Z_s| + C |\bar K_s| \sqrt \lambda_s + |\bar f_s|.
$$
%By using the well known inequalities $2ab \leq a^2 + b^2$ and $(a+b+c)^2 \leq 3(a^2+b^2+c^2)$, 
Note now that,  for all non negative numbers $\lambda$, $y$, 
$z$, $k$, $f$ and $\varepsilon >0$, we have\\
$ 2y (Cz + Ck \sqrt  \lambda  + f) \leq \frac{ y^2}{\varepsilon^2}+ \varepsilon^2(Cz+ Ck  \sqrt \lambda   + f)^2 \leq \frac{ y^2}{\varepsilon^2} +  3 \varepsilon^2(C^2 y^2+ C^2 k^2\lambda +f^2)$. Hence, 
%$ 2x (C(\pi + l) + t) \leq C \frac{(\pi + l)^2}{\varepsilon^2} + \frac{t^2}{\mu^2} + x^2 (\mu^2 + C \varepsilon^2) \; \text{ and } \; (\pi + l)^2 \leq 2( Z^2 + l^2),$$ 
%we get
\begin{align}\label{eq2a}
&e^ {\beta t}  \bar Y_t^2  +  {\mathbb E} \left[\beta \int_t^T e^{\beta s}  \bar Y_s^2 ds + \int_t^T e^{\beta s} ( \bar Z_s^2 +
 \bar K_s^2\lambda_s) ds \mid {\cal G}_t \right] \nonumber \\
%& \leq {\mathbb E} \left[ e^{\beta T} ( \xi_1 - \xi_2)^2 + \int_t^T e^{\beta s} \left([C(2 + \varepsilon^2) + \mu^2]  \bar Y_{\cal S}^2  + C \frac{( \bar Z_s + \bar k_s)^2}{ \varepsilon^2} + \frac{  \bar f_{\cal S}^2}{\mu^2}  \right) ds \mid \FC_t \right] \nonumber \\
& \leq {\mathbb E} \left[ (2C+\frac{ 1}{\varepsilon^2}) \int_t^T  e^{ \beta s}  \bar Y_s^2 ds + 3C^2 \varepsilon^2 \int_t^T e^{\beta s} ( \bar Z_s^2 + \bar K_s^2\lambda_s)ds
  + 3 \varepsilon^2   \int_t^T e^{\beta s}  \bar f_s^2 ds \mid {\cal G}_t \right].
\end{align}
Let us make the change of variable $\eta = 3 \epsilon^2$. Then, for each  $\beta,  \eta>0$  chosen as in the proposition,  these inequalities lead to
\eqref{A26}. 
By integrating  \eqref{A26}, we obtain $\|\bar Y \|_\beta^2 \leq T  \eta
\|\bar f \|_\beta^2$. Using  inequality \eqref{eq2a}, the last assertion of the Proposition  follows.
\end{proof}
From these estimates, we derive an existence and uniqueness result for DRBSDEs.\\
{\bf Proof of the existence and the uniqueness of the solution of a DRBSDE \eqref{DRBSDE}}: 
Let us first consider the case when the driver $g(t)$ does not depend on the solution. 
By using the representation property of ${\mathbb G}$-martingales (Lemma \ref{theoreme representation}) 
and some results of Dynkin games theory, one can show, proceeding as in \cite{DQS2}, that there exists a unique solution of the associated 
DRBSDE \eqref{DRBSDE}. 
The proof in the general case is the same as for non reflected BSDEs with default jump (see the proof of Proposition 2.6 in \cite{DQS4}). It is based on a fixed point argument, using the previous estimates. $\square$

We state a comparison result of  analysis for differential equations (deterministic).
Let $L^2=L^2([0,T], dt)$ be the space of square integrable Borelian  real valued maps on $[0,T]$.
\begin{lemma}[Comparison  for differential equations] \label{classique}
For $i=1,2$, let $b_i: [0,T] \times {\mathbb R} \rightarrow {\mathbb R}; (t,y) \mapsto b_i(t,y)$ be a Borelian map with 
$b_i(.,0)$ $\in$ $L^2$, and supposed to be uniformly Lipschitz with respect to $y$.
%such that for all $t \in [0,T]$ and $y_1,y_2$ $\in {\mathbb R}$, we have
%\begin{equation}\label{hypb}
%(b_i(t, y_1) - b_i(t, y_2))(y_1 - y_2)\geq
%-C |y_1 - y_2| ^2 ,
%\end{equation}
%with Lipschitz constant $C>0$.
Let $f^1$, $f^2$ be right-continuous maps in $L^2$ and let 
 $x_1, x_2$ $\in {\mathbb R}$.
For $i=1,2$, let  $y^i$  be the unique right-continuous map  in $L^2$ satisfying  the differential equation:
$$y^i_t := x_i + \int_0^t b_i(s, y^i_s)ds + f^i(t).$$
Suppose that $x_1 \geq x_2$ and $b_1(t, y^2_t)\geq b_2(t, y^2_t)$ $0 \leq t \leq T$ $ds$-a.e.  Suppose also that 
 $f^1= f^2 + A$, where $A$ is a 
non decreasing right-continuous 
map on $[0,T]$ with $A_0=0$.
 We then have $y^1_t\geq y^2_t$ for each $t \in [0,T]$. 
 Moreover, if $x_1 > x_2$, then $y^1_t > y^2_t$ for each $t \in [0,T]$. 
\end{lemma}

\begin{proof} 
The proof is classical.
We have $d\bar y_t= \lambda _t \bar y_t dt +(b_1(t, y^2_t)- b_2(t, y^2_t))dt+  dA_t$, where $\bar y:= y^1-y^2$
 and 
$\lambda _t:= ( b_1(t, y^1_t)- b_1(t, y^2_t)) (\bar y_t)^{-1} {\bf 1}_{ y^1_t\neq y^2_t}$. 
%Note that by the assumption made on $b$, 
%$- \lambda \leq C$.
Hence, 
$\bar y_t= x_1- x_2+ \int_0^t e^{-\int_s ^t \lambda_u du}(b_1(s, y^2_s)- b_2(s, y^2_s))ds + \int_0^t e^{-\int_s ^t \lambda_u du} dA_s \geq 0$. Moreover, if $x_1 > x_2$, then the inequality is strict.
\end{proof}

\end{document}